\documentclass[12pt]{article}
\usepackage{amsmath}
\usepackage{graphicx,psfrag,epsf}
\usepackage{enumerate}
\usepackage{natbib}
\usepackage{url} 

\newcommand{\blind}{0}

\usepackage{adjustbox}
\usepackage{graphicx}
\DeclareOldFontCommand{\bf}{\normalfont\bfseries}{\mathbf}
\usepackage[format = hang]{subcaption}
\usepackage[format = hang]{caption}
\usepackage{pgfplots}
\usepackage{booktabs}
\usepackage{array}
\usepackage{float}
\usepackage{enumitem}
\usepackage[normalem]{ulem}
\usepackage{lscape}
\usepackage{amsmath, amssymb, amsthm, bbm, mathtools}
\allowdisplaybreaks
\usepackage{commath}
\usepackage{thmtools}
\declaretheoremstyle[
headfont=\bfseries,%
headpunct={.}, %
numbered=yes,
spaceabove=10pt, %
postheadspace=10pt ]{bettertheorem}
\declaretheorem[name=Theorem,style=bettertheorem]{theorem}
\declaretheorem[name=Corollary,style=bettertheorem]{corollary}
\declaretheorem[name=Proposition,style=bettertheorem]{proposition}
\declaretheorem[name=Lemma,style=bettertheorem]{lemma}
\declaretheorem[name=Assumption,style=bettertheorem]{assumption}
\declaretheorem[name=Definition,style=bettertheorem]{definition}
\declaretheoremstyle[
headfont=\bfseries,%
headpunct={.}, %
numbered=no,
spaceabove=10pt, %
postheadspace=10pt ]{pptheorem}

\usepackage{mathrsfs}
\usepackage{listings}

 \newcolumntype{P}[1]{>{\centering\arraybackslash}p{#1}}
 \newcolumntype{L}[1]{>{\raggedright\arraybackslash}p{#1}}
\newcolumntype{M}[1]{>{\centering\arraybackslash}m{#1}}

\usepackage{authblk}

\usepackage[title]{appendix}

\usepackage[framemethod=tikz,rightmargin=5,
leftmargin=5,backgroundcolor=gray!10,
frametitlerule=true,frametitlebackgroundcolor=gray!30,roundcorner=10pt]
{mdframed}

\begin{document}

\def\spacingset#1{\renewcommand{\baselinestretch}%
{#1}\small\normalsize} \spacingset{1}


\if0\blind
{
  \title{\bf A Ridge-Regularised Jackknifed Anderson-Rubin Test}
  \author{Max-Sebastian Dov\`i \thanks{
	This research is funded by the European Research Council via Consolidator grant number 647152 and the German National Merit Foundation. The views expressed herein are those of the authors and should not be attributed to the IMF, its Executive Board, or its management. We thank Frank Kleibergen, Damian Kozbur, Anna Mikusheva, Bent Nielsen and Alexei Onatski for useful comments and discussions.}\hspace{.2cm}\\
    International Monetary Fund\\
    and \\
    Anders Bredahl Kock \\
    Department of Economics, University of Oxford and CREATES, Aarhus University\\
    and\\
    Sophocles Mavroeidis\\
    Department of Economics, University of Oxford}
  \maketitle
} \fi

\if1\blind
{
  \bigskip
  \bigskip
  \bigskip
  \begin{center}
    {\LARGE\bf A Ridge-Regularised Jackknifed Anderson-Rubin Test}
\end{center}
  \medskip
} \fi

\bigskip
\begin{abstract}
We consider hypothesis testing in instrumental variable regression models with few included exogenous covariates but many instruments -- possibly more than the number of observations. We show that a ridge-regularised version of the jackknifed \citet[][henceforth AR]{Anderson:1949tx} test controls asymptotic size in the presence of heteroskedasticity, and when the instruments may be arbitrarily weak. Asymptotic size control is established under weaker assumptions than those imposed for recently proposed jackknifed AR tests in the literature. Furthermore, ridge-regularisation extends the scope of jackknifed AR tests to situations in which there are more instruments than observations. Monte-Carlo simulations indicate that our method has favourable finite-sample size and power properties compared to recently proposed alternative approaches in the literature. An empirical application on the elasticity of substitution between immigrants and natives in the US illustrates the usefulness of the proposed method for practitioners.
\end{abstract}

\noindent%
\newpage
\spacingset{1.8} 

\section{Introduction}

Instrumental variables (IVs) are commonly employed in economics and related fields to estimate causal effects from observational data. 
 The use of a large number of IVs has gained popularity due to weak identification, where researchers aim to capture limited exogenous variation in endogenous covariates and obtain more precise inference. One example is Mendelian randomisation in biology, where weakly associated genetic mutations are used as IVs, and the number of IVs can exceed the number of observations \citep{Davies:2015,burgess2021mendelian}. Other prominent examples are the granular IV approach \citep{Gabaix:2020ws} and the saturation approach to identify treatment effects \citep{Blandhol2022}.\footnote{Recent examples of applications with many IVs include \citet{Muller}, who uses judge dummies interacted with individual covariates as IVs to identify causal effects of incarceration; \citet{vanDuijn}, who derive an equation for residential sorting that produces many more IVs than observations; and \citet{Lonn}, who derive structural asset-pricing models that lead to a vast number of moment conditions.} However, it is crucial to have reliable inference methods that account for the potential lack of joint informativeness of a large number of weak IVs, especially when pretesting the strength of IVs is impractical or not available due to the excessive number of IVs compared to observations. The present paper contributes to the literature on the development of such methods.

We propose a ridge-regularised jackknifed Anderson-Rubin (RJAR) test to construct confidence sets for the coefficients of endogenous variables in weakly-identified and heteroskedastic IV models when the number of IVs is large. Jackknife-based methods have recently been used in this context because they are applicable in an asymptotic framework where the number of IVs diverges with the number of observations. However, by relying on existing central limit theorems developed for standard projection matrices, these methods require that the number of IVs be less than the number of observations, and often perform poorly when the number of IVs is close to (but still less than) the number of observations. Recently-proposed regularisation approaches for inference under many IVs require strong identification or a sparse relationship between the endogenous variables and the IVs to work well. By combining jackknifing with ridge regularisation, we provide a test that has the desired asymptotic size under heteroskedasticity, arbitrarily weak identification, and more IVs than observations, all while achieving good power both when the relationship between the endogenous variables and the IVs is sparse and when it is dense.

Throughout this paper, we focus on \citet[henceforth AR]{Anderson:1949tx} tests. This is because, in addition to being robust to weak IVs, they are also robust to arbitrary relationships between the endogenous variables and the IVs, since they make no assumption on the first-stage projection of the endogenous variables on the IVs. In the context of AR tests with an increasing number of IVs, it is helpful to distinguish three different asymptotic regimes. 

The first `moderately many' IVs regime allows the number of IVs, $k_n$, to grow with the sample size, $n$, but still requires it to be asymptotically negligible with respect to the sample size. Examples include \citet{Andrews:2007bl}, who require $k_n/n^{1/3} \to 0$, and \citet{Phillips:2017bz}, who require $k_n / n \to 0$.

The second `many' IVs regime allows the number of IVs to be of the same order of magnitude as the number of observations. \citet{Anatolyev:2010kk} provide an AR test that remains valid for $k_n/n \to \tau$, $0 \leq \tau < 1$, $n \to \infty$, provided that the error terms are homoskedastic and the IVs satisfy a restrictive balanced-design assumption.\footnote{See \citet{Anatolyev:2017ee} and \citet{Crudu:2020bk} for a discussion of the restrictiveness of the balanced-design assumption of \citet{Anatolyev:2010kk}.}$^,$\footnote{See \citet{Kaffo:2017gp} for a bootstrapped version of the \citet{Anatolyev:2010kk} AR test.} \citet{Bun:2021en} provide analogous results for the GMM version of the AR statistic under the assumption of independently and identically distributed (i.i.d.) data. These results were recently extended in the form of a jackknifed AR statistic by \citet[henceforth CMS]{Crudu:2020bk} and \citet[henceforth MS]{Mikusheva:2020uo} to the case where errors are allowed to display arbitrary heteroskedasticity, and the only assumption on the IVs is that the diagonal entries of the projection matrix of the IVs are bounded away from unity from above.\footnote{ \citet{Boot, Matushita:2022,Lim} propose weak-identification robust jackknife-based procedures for inference with many IVs. However, these methods make additional assumptions on the first-stage projection of the endogenous variables on the IVs, and hence fall outside of the scope of this paper.}

The third `very many' IVs regime allows the number of IVs to grow with $n$, and further allows the number of IVs to be greater than the number of observations. \citet[henceforth BCCH]{Belloni:2012kw} propose a Sup Score test that remains valid under mild conditions, and allows the number of IVs to increase exponentially with the sample size. \citet[henceforth CT]{Carrasco:2016un} propose a ridge-regularised AR statistic that allows for more IVs than observations under the assumption of i.i.d.~data and homoskedasticity. \citet{Kapetanios:2015fp} extend \citet{Bai:2010ds} by proposing weak-identification robust factor-based tests that in principle allow for the number of IVs to be larger than the number of observations, provided that the factor structure of the IVs is sufficiently strong.

Using the notation of the model introduced in the next section, Table \ref{Table-Frontier} provides a schematic overview of the main assumptions and results in the literature on inference that is robust to many weak IVs.

Our test provides a twofold extension of the existing literature. First, it allows for valid inference under many IVs and heteroskedastic errors while further weakening the assumptions of similar tests proposed by CMS and MS. This is made possible by deriving the limiting behaviour of the RJAR statistic from the bottom up, without relying on the existing results in \citet{Chao:2012iz} or \citet{Hansen:2014ie}. Second, this test allows for more IVs than observations under arbitrary heteroskedasticity of the error terms. The only other approach currently available in the literature that is robust to heteroskedastic error terms and more IVs than observations is the Sup Score test of BCCH. Simulations show that the RJAR test has power comparable to the Sup Score test of BCCH whenever the signal in the first stage is sparse (i.e., only a few of the IVs are informative), and substantially more power when the signal in the first stage is dense (i.e., not sparse). Simulations also show that the RJAR test is as powerful as, and in some cases more powerful than, the AR test of CT, the validity of which has been established only under homoskedastic error terms. Finally, we provide a comparison of the most recently proposed approaches to conducting inference in possibly heteroskedastic linear IV models when the number of IVs is not negligible with respect to the sample size. Indeed, using extensive simulation evidence and an empirical application based on \citet{Card:2009hj}, we provide a comparison between these existing `state-of-the-art' methods in a controlled and comparable setting.

In the rest of the paper, we use the following notation. $I_p$ denotes the $p \times p$ identity matrix. For an $m\times 1$ vector $a = [a_1, \dots, a_m]'$, $||a||_2 := \sqrt{\sum_{j = 1}^{m}a_j^2}$. The entry $(i, j)$ for an $m\times p$ matrix $A$ is denoted as $A_{ij}$ for $1\leq i \leq m$, $1 \leq j \leq p$.  $\text{tr}(A):= \sum_{i=1}^mA_{ii}$ for any $m\times m$ matrix $A$ with entries given by $A_{ij}$ for $1 \leq i, j \leq m$, and $||A||_F\coloneqq \sqrt{\text{tr}(A'A)}$ denotes the Frobenius norm of $A$. The remaining notation follows standard conventions.

The structure of the paper is as follows. Section \ref{Section-RJAR-Model} specifies the linear IV model considered throughout. Section \ref{Section-RJAR-HDAR} introduces the RJAR test, and provides the main asymptotic results. Section \ref{Section-RJAR-Simulations} provides simulation evidence on the size and power of our RJAR test and compares it with the tests proposed by BCCH, CT, CMS and MS. Section \ref{Section-RJAR-App} considers an empirical application based on \citet{Card:2009hj}. Section \ref{Section-RJAR-Conclusion} concludes. All proofs are given in the Supplementary Material.

\begin{table}[H]
\scriptsize
\caption{\citet{Anderson:1949tx} tests with many IVs: schematic comparison of main assumptions and results in the literature.}\label{Table-Frontier}
\centering
\begin{tabular}{P{0.22\textwidth}P{0.2\textwidth}P{0.2\textwidth}P{0.15\textwidth}}
\toprule
 & $k_n$ &  $Z$ & $\varepsilon$ \\
\midrule
%

\citet{Kaffo:2017gp, Anatolyev:2010kk} & $k_n/n \to \tau$, $0 \leq \tau < 1$, $n \to \infty$  & $Z$ is fixed, obeys restrictive balanced design assumption & i.i.d., $\mathbb{E}[\varepsilon_i] = 0$, $\mathbb{E}[\varepsilon_i^4] < \infty$ \\

\\

\citet{Belloni:2012kw} & $\log(\text{max}(k_n,n)) = o(n^{1/3})$, $n \to \infty$  & $Z$ is fixed & independent, $\mathbb{E}[\varepsilon_i] = 0$, $\mathbb{E}[\varepsilon_i^3] < \infty$ \\

\\

\citet{Kapetanios:2015fp} & Factor-dependent  & \multicolumn{2}{P{0.35\textwidth}}{Potentially dependent data, $\mathbb{E}[Z_{il}^4] < \infty$, $l = 1, \dots, k_n$, $\mathbb{E}[\varepsilon_i^4] < \infty$ }\\

\\


\citet{Carrasco:2016un} & $n \to \infty$ & \multicolumn{2}{P{0.35\textwidth}}{$Z_i\varepsilon_i$ is i.i.d., $\mathbb{E}[Z_{i}\varepsilon_i] = 0$, $\mathbb{E}[X_{ig}^2] < \infty$} \\

\\

%

\citet{Bun:2021en} & $k_n/n \to \tau$, $0 \leq \tau < 1$, $n \to \infty$  & \multicolumn{2}{P{0.35\textwidth}}{$Z_i\varepsilon_i$ is i.i.d., $Z$ obeys restrictive balanced design assumption, $\mathbb{E}[Z_i\varepsilon_i] = 0$, $\mathbb{E}[Z_i^8\varepsilon_i^8] < \infty$} \\

\\

\citet{Crudu:2020bk} & $k_n/n \to \tau$, $0 \leq \tau < 1$, $k_n \to \infty$ as $n\to \infty$  & $Z$ is fixed, $P_{ii} \leq 1-\delta$, $0 <\delta < 1$ & independent, $\mathbb{E}[\varepsilon_i] = 0$, $\mathbb{E}[\varepsilon_i^4] < \infty$ \\

\\

\citet{Mikusheva:2020uo} & $k_n/n \to \tau$, $0 \leq \tau < 1$, $k_n \to \infty$ as $n\to \infty$ & $Z$ is fixed, $P_{ii} \leq 1-\delta$, $0 <\delta < 1$ & independent, $\mathbb{E}[\varepsilon_i] = 0$, $\mathbb{E}[\varepsilon_i^6] < \infty$ \\

\midrule

RJAR & $r_n \to \infty$ as $n\to \infty$  & $Z$ is fixed, $\underset{n\to \infty}{\liminf} \frac{1}{r_n}\sum_{i = 1}^n\sum_{j\neq i}(P^{\gamma_n}_{ij})^2 > 0$ & independent, $\mathbb{E}[\varepsilon_i] = 0$, $\mathbb{E}[\varepsilon_i^4] < \infty$ \\

\bottomrule
\multicolumn{4}{l}{\emph{Notes}:}\\
\multicolumn{4}{p{0.93\textwidth}}{$n$ is the number of observations, $k_n$ is the number of IVs, and $g$ is the number of endogenous variables. $Z$ is the $n\times k_n$ matrix of IVs. $X$ is the $n\times g$ matrix of endogenous variables. $\varepsilon$ is the $n\times 1$ vector of structural error terms. $r_n := \text{rank}(Z)$. See also Equation \eqref{Equation-RJAR-Model} below.}\\
\multicolumn{4}{p{0.93\textwidth}}{$P := Z(Z'Z)^{-1}Z'$, $P^{\gamma_n} := Z(Z'Z + \gamma_n I_{k_n})^{-1}Z'$ for $\gamma_n  \geq 0$ if $r_n = k_n$ and $\gamma_n \geq \gamma_{-} > 0$ if $r_n < k_n$.}\\
\end{tabular}
\end{table}

\section{Model \label{Section-RJAR-Model}}

We consider the heteroskedastic linear IV model given by
\begin{subequations}
	\label{Equation-RJAR-Model}
	\begin{align}
		y &= X\beta + \varepsilon \label{Equation-RJAR-Model a}\\
		X &= Z\Pi + V, \label{Equation-RJAR-Model b}
	\end{align}
\end{subequations}
 where $y$ is an $n\times 1$ vector containing the dependent variable, $X$ is an $n\times g$ matrix containing the endogenous variables, $\beta$ is a $g\times 1$ coefficient vector, $\varepsilon$ is an $n\times 1$ vector containing the structural error terms, $Z$ is an $n \times k_n$ matrix containing the IVs, $\Pi$ is a $k_n\times g$ coefficient matrix, and $V$ is an $n\times g$ matrix containing the first-stage errors. $k_n$ can diverge with $n$, but $g$ is fixed. Also, let $y_i$, $X_i$, $\varepsilon_i$, $Z_i$, and $V_i$ denote the $i^{\text{th}}$ row of $y, X, \varepsilon, Z$, and $V$, respectively. As in BCCH, CMS and MS, we treat $Z$ as fixed (non-stochastic),\footnote{The treatment of IVs as non-stochastic is an assumption that has been commonly used in several of the papers that are most closely related to ours, see Table 1. 
    With stochastic IVs, one can interpret the results as holding conditional on $Z$, see also CMS Footnote 1.} and assume that any exogenous covariates have been partialled out (see the discussion of Assumption \ref{Assumption-High-Level} below). We \emph{exclusively} consider methods that allow for arbitrarily weak identification, i.e., methods that control asymptotic size irrespective of the value of $\Pi$.


Inference is conducted on the coefficient vector $\beta$ by testing hypotheses of the following type for a prespecified $\beta_0 \in \Re^g$:
\begin{equation}
	\label{Equation-RJAR-Null-Hypothesis}
	H_0: \beta = \beta_0 \text{ vs. } H_1: \beta \neq \beta_0.
\end{equation}

For a given non-randomised test of asymptotic size $\alpha \in (0, 1)$, a confidence set of asymptotic coverage $1-\alpha$ can be constructed as the collection of those~$\beta_0$ for which $H_0$ in Equation \eqref{Equation-RJAR-Null-Hypothesis} is not rejected. For convenience, define 
\begin{equation}\label{eq: struct error under H0}
    e(\beta_0) \coloneqq \left[e_1(\beta_0), \ldots, e_n(\beta_0) \right]', \qquad e_i(\beta_0) \coloneqq  y_i - X_i\beta_0, \qquad i = 1, \dots, n,
\end{equation}
which we refer to as the structural error under the null hypothesis.


It will become apparent below that so long as the error term, $\varepsilon$, remains additive, the linearity in the structural equation \eqref{Equation-RJAR-Model a} and the first-stage equation \eqref{Equation-RJAR-Model b} can be relaxed to allow for any (known) real-valued function, without affecting the asymptotic size of the RJAR test.




\section{The RJAR test \label{Section-RJAR-HDAR}}

This section introduces our proposed RJAR test, derives its large sample properties under the null hypothesis in Equation \eqref{Equation-RJAR-Null-Hypothesis}, and discusses its relationship to the most closely related tests in the literature: the ridge-regularised AR test of CT, the jackknifed AR tests of CMS and MS, and the Sup Score test of BCCH. 

\subsection{Definition of the RJAR test}

The original AR test of the null hypothesis in Equation \eqref{Equation-RJAR-Null-Hypothesis} can be thought of as a test that the IVs are exogenous using the implied structural errors in Equation \eqref{eq: struct error under H0}. More specifically, the AR test is a Wald test of the significance of the IVs in the auxiliary regression of the structural errors under the null, $e_i(\beta_0)$, on the IVs, $Z_i$.\footnote{The AR statistic is defined as $e(\beta_0)'Z \widehat{\Sigma}^{-1} Z'e(\beta_0)/n$, where $\widehat{\Sigma}$ is a consistent estimator of $\text{Var}[Z'e(\beta_0)/\sqrt{n}]$.} The weak-IV-robust AR test of the null hypothesis in Equation \eqref{Equation-RJAR-Null-Hypothesis} with asymptotic size $\alpha$ rejects the null if and only if the AR statistic exceeds the $1-\alpha$ quantile of a $\chi^2$ distribution with $k_n$ degrees of freedom. When the number of IVs $k_n$ grows with $n$, i.e., when there are many (potentially weak) IVs, the $\chi^2$ approximation of the original AR statistic breaks down. The recently proposed AR test of CT, which we shall briefly review in Section \ref{Section-RJAR-Competition}, uses ridge-regularisation to allow for a growing number of IVs that may even exceed the sample size, but is not robust to arbitrary heteroskedasticity in the error terms. The recent papers by CMS and MS, which we shall also briefly review in Section \ref{Section-RJAR-Competition}, propose jackknifed versions of the AR test that remain valid when $k_n$ grows with $n$, but require $k_n<n$. Our proposed method combines ridge regularisation with jackknifing of the auxiliary regression of $e_i(\beta_0)$ on $Z_i$ to allow for a weak-identification robust test capable of dealing with both more IVs than observations and arbitrary heteroskedasticity in the error terms.

We first standardise the IVs in-sample (after partialling out any covariates) as in BCCH (p.~2393) so that
\begin{equation}
    \label{Equation-Standardised}
    \frac 1 n \sum_{i = 1}^n Z_{ij}^2 = 1 \quad \text{for} \quad j = 1, \dots, k_n.
\end{equation}

The RJAR test for the testing problem in Equation \eqref{Equation-RJAR-Null-Hypothesis} is then based on the following statistic:
\begin{align}\label{eq: RJAR}
	RJAR_{\gamma_n}(\beta_0) & \coloneqq   \frac{1}{\sqrt{r_n} \sqrt {\hat{\Phi}_{\gamma_n}(\beta_0)}}\sum_{i = 1}^n \sum_{j \neq i}P_{ij}^{\gamma_n}e_i(\beta_0)e_j(\beta_0),
\end{align}
where $r_n\coloneqq \text{rank}(Z)$ (assumed to be positive without loss of generality), 
\begin{equation}\label{eq: variance estimator}
\hat{\Phi}_{\gamma_n}(\beta_0) \coloneqq  \frac{2}{r_n}\sum_{i = 1}^n \sum_{j \neq i}(P_{ij}^{\gamma_n})^2e_i^2 (\beta_0)e_j^2(\beta_0)    ,
\end{equation}
and $P^{\gamma_n} := Z(Z'Z + \gamma_n  I_{k_n})^{-1}Z'$ is the ridge-regularised projection matrix with $\gamma_n \geq 0$ if $r_n = k_n$ and $\gamma_n>0$ if $r_n < k_n$. $\gamma_n$ is a (sequence of) regularisation parameter(s) whose choice we discuss below. We note that setting $\gamma_n = 0$ when $r_n = k_n$ corresponds to the case where a standard least-squares projection matrix $Z(Z'Z)^{-1}Z'$ is used to construct the RJAR statistic. This is the standard jackknife statistic. We also note that the $r_n$-scaling in the denominator of the RJAR statistic cancels with the $r_n$ in the denominator of $\hat{\Phi}_{\gamma_n}(\beta_0)$. 


We set the regularisation parameter $\gamma_n$ to
	\begin{equation}
		\label{Equation-Gamma-Star}
		\gamma^*_n \coloneqq \max\underset{\gamma_n  \in \Gamma_n}{\arg\max}\sum_{i = 1}^n\sum_{j\neq i}(P^{\gamma_n}_{ij})^2,
	\end{equation}
where $\Gamma_n \coloneqq \{\gamma_n\in \Re : \gamma_n \geq 0 \text{ if } r_n = k_n, \text{ and } \gamma_n \geq \gamma_{-} > 0 \text{ if } r_n<k_n\}$ for some constant $\gamma_{-}$ not depending on $n$. The existence of $\gamma^*_n$ is shown in Lemma \ref{Lemma-Gamma-Star-Existence} in the Supplementary Material. We let $\gamma_n^*$ be an element of $\underset{\gamma_n  \in \Gamma_n}{\arg\max}\sum_{i = 1}^n\sum_{j\neq i}(P^{\gamma_n}_{ij})^2$ out of conservativeness, i.e., to make Assumption \ref{Assumption-High-Level} below as plausible as possible given the IVs. Furthermore, we let $\gamma^*_n$ be the maximal element of this set because the maximiser is not necessarily unique without imposing additional assumptions on the singular values and left-singular vectors of the IVs (although in practice we only found unique maximisers). We choose to take the maximum of the maximisers to make the smallest eigenvalues of the ridge-regularised Gram matrix, $Z'Z + \gamma_nI_{k_n}$, as far away from zero as possible when $r_n < k_n$. To see that~$\gamma_n^*$ maximises the smallest eigenvalue of~$Z'Z + \gamma_nI_{k_n}$ among the elements of~$\underset{\gamma_n  \in \Gamma_n}{\arg\max}\sum_{i = 1}^n\sum_{j\neq i}(P^{\gamma_n}_{ij})^2$, notice that due to the symmetry of $Z'Z + \gamma_nI_{k_n}$, its smallest eigenvalue can be expressed as
\begin{equation*}
    \underset{a\in \Re^{k_n}:||a||_2 = 1}{\text{ min }} a'(Z'Z + \gamma_nI_{k_n})a = \underset{a\in \Re^{k_n}:||a||_2 = 1}{\text{ min }} a'(Z'Z)a + \gamma_n = \gamma_n,
\end{equation*}
where the first equality used that $a'a = 1$, and the second equality used that the smallest eigenvalue of $Z'Z$ is 0 when the rank $r_n$ of $Z$ is less than $k_n$.

We now have all the ingredients to define our new test.
\begin{definition}\label{def: RJAR}
The RJAR test rejects $H_0:\beta=\beta_0$ in Equation \eqref{Equation-RJAR-Null-Hypothesis} at significance level $\alpha\in(0,1)$ if and only if
\begin{equation}\label{eq: RJAR test}
RJAR_{\gamma^*_n}(\beta_0) > \mathcal{Q}(1-\alpha),    
\end{equation}
where $RJAR_{\gamma_n}(\beta_0)$ is defined in Equation \eqref{eq: RJAR}, $\gamma^*_n$ is defined in Equation \eqref{Equation-Gamma-Star}, and $\mathcal Q(1-\alpha)$ is the $(1-\alpha)$ quantile of the Standard Normal distribution.
\end{definition}

\subsection{Asymptotic properties of the RJAR test \label{Section-Asymptotics}}

We make the following assumptions to derive the limiting distribution of $RJAR_{\gamma^*_n}(\beta_0)$ under the null hypothesis in Equation \eqref{Equation-RJAR-Null-Hypothesis}.

\begin{assumption}
	\label{Assumption-MC}
	$\{\varepsilon_i\}_{i\in\mathbb{N}}$ is a sequence of independent random variables satisfying $\mathbb{E}[\varepsilon_i] = 0$, $\inf_{i\in \mathbb{N}}\text{Var}[\varepsilon_i]> 0$, and $\sup_{i\in\mathbb{N}}\mathbb{E}[\varepsilon_i^4] < \infty$.
\end{assumption}

\begin{assumption}
	\label{Assumption-Bekker}
	$r_n = \text{rank}(Z) \to \infty$ as $n \to \infty$. 
\end{assumption}


\begin{assumption}
	\label{Assumption-High-Level}
\hfill \break
\vspace{-1.0cm}
	\begin{enumerate}
	    \item If $r_n = k_n$, then there exists a $\gamma_n\in [0, \infty)$ such that $\underset{n\to \infty}{\liminf} \frac{1}{r_n}\sum_{i = 1}^n\sum_{j\neq i}(P^{\gamma_n}_{ij})^2 > 0$.
	    \item If $r_n < k_n$, then there exists a~$\gamma_{-} > 0$  and a~$\gamma_n\in [\gamma_{-}, \infty)$ such that\\ $\underset{n\to \infty}{\liminf} \frac{1}{r_n}\sum_{i = 1}^n\sum_{j\neq i}(P^{\gamma_n}_{ij})^2 > 0$.
	\end{enumerate}
\end{assumption}

Assumption \ref{Assumption-MC} is a mild condition on the structural error terms, and allows for conditional heteroskedasticity. It is the same as in CMS. It is slightly less restrictive than the one in MS (who require finite sixth moments on the structural error terms), and slightly more restrictive than the one in BCCH (who require finite third moments).\footnote{It is difficult to allow for autocorrelated errors in our framework, since our proof relies on casting the RJAR statistic as a degenerate U-statistic, which requires that $\varepsilon_i$ be independent of $\varepsilon_j$ for $i\neq j$.}

Assumption \ref{Assumption-Bekker} is a weak technical assumption. It implies that both $k_n$ and $n$ diverge. It also allows for the sum of the number of IVs and the number of exogenous covariates to be larger than the number of observations, provided that the number of exogenous covariates that have been partialled out be sufficiently small (so that the rank of the matrix of partialled IVs continues to diverge).\footnote{{More formally, suppose that $W$ is an $n\times l_n$ matrix of covariates with rank $l_n$, and define $M_W := I_{n} - W(W'W)^{-1}W'$. Furthermore, let $\tilde{Z}$ be the $n\times k_n$ matrix of IVs prior to partialling out the controls $W$, and let $\tilde{r}_n := \text{rank}(\tilde{Z})$, such that  $Z = M_W\tilde{Z}$ and $r_n := \text{rank}(Z) \leq \text{min}(n-l_n, \tilde{r}_n)$. Therefore, $l_n$ can diverge as $n\to \infty$ without contradicting $r_n\to \infty$, which means that a growing number of exogenous controls can be accommodated.}} This assumption is weaker than the restriction on the dimensionality in MS and CMS, who require $r_n = k_n$, and $k_n < n$ for each $n\in \mathbb N$, as well as $k_n\to \infty$ as $n\to \infty$. BCCH prove asymptotic size control of their Sup Score test under the assumption that $\log k_n = o(n^{1/3})$.


\subsubsection{Assumption \ref{Assumption-High-Level} when~$k_n<n$}\label{Section: A3 k<n}
Assumption \ref{Assumption-High-Level} is a high-level assumption on the number of IVs and their correlation structure. Assumption \ref{Assumption-High-Level} implies that $\gamma^*_n$ satisfies $\underset{n\to \infty}{\liminf} \frac{1}{r_n}\sum_{i = 1}^n\sum_{j\neq i}(P^{\gamma^*_n}_{ij})^2 > 0$. In the absence of exogenous covariates and when $k_n < n$ (as in CMS and MS), Assumption \ref{Assumption-High-Level} is weaker than the \emph{balanced-design} assumption in CMS and MS which requires for $P := Z(Z'Z)^{-1}Z'$ that
\begin{equation}
    \label{Equation-CMS-MS-BD}
    \underset{1 \leq i \leq n}{\text{ max }}P_{ii} \leq 1 - \delta\qquad\text{for all }n\in\mathbb{N},
\end{equation}
for some $0 < \delta < 1$.  This is stated formally in part 1 of the following Proposition.
\begin{proposition}
	\label{Proposition-Weaker-CMSMS}
	Suppose (as in CMS and MS) that $\text{rank}(P) = k_n$. 
 \begin{enumerate}
     \item If there exists a $\delta\in(0,1)$ such that $\underset{1 \leq i \leq n}{\text{ max }}P_{ii} \leq 1 - \delta$ for all~$n\in\mathbb{N}$ then Assumption \ref{Assumption-High-Level} is satisfied.
     \item For any~$\delta\in(0,1)$, let~$\mathcal{A}_n(\delta):=\{i \in \{1, \dots, n\}: P_{ii} \geq 1 - \delta\}$. If there exists a~$\delta'\in(0,1)$ such that~$\frac{1}{k_n}|\mathcal{A}_n(\delta')| \to 0$ then Assumption \ref{Assumption-High-Level} is satisfied.
 \end{enumerate}
\end{proposition}
Part 2.~of Proposition \ref{Proposition-Weaker-CMSMS} implies that Assumption \ref{Assumption-High-Level} can be satisfied --- even when the balanced-design assumption is not --- as long as the number of diagonal elements of~$P$ ``close to 1'' increases slower than~$k_n$. Note also that if for some~$\delta'\in(0,1)$ it holds that $\underset{1 \leq i \leq n}{\text{ max }}P_{ii} \leq 1 - \delta'$ for all~$n\in\mathbb{N}$, then~$\mathcal{A}_n(\delta')=\emptyset$. Thus, part 1.~of Proposition~\ref{Proposition-Weaker-CMSMS} actually follows from Part 2. However, we have chosen to keep part 1.~as a separate statement to explicitly state that the balanced-design assumption implies Assumption \ref{Assumption-High-Level}.

\subsubsection{Assumption \ref{Assumption-High-Level} when~$k_n\geq n$}\label{Section:A3 with k>n}
In case~$k_n\geq n$ the following proposition provides sufficient conditions for Assumption~\ref{Assumption-High-Level} to be satisfied with probability one in case~$Z$ is random.
\begin{proposition}\label{prop:hd}
    Let the entries of~$Z$ be i.i.d.~with a mean zero and variance one. In addition, let~$p>2$ satisfy~$\mathbb{E}|Z_{11}|^{4p}<\infty$ and let~$k_n=\tau n$ for~$\tau\in[1,\infty)$. Then, if the distribution of~$Z_{11}$ is absolutely continuous with respect to the Lebesgue measure, there exists a constant~$\eta>0$ such that for~$\gamma_{n,\eta}=\eta  n$ it holds with probability one that
    \begin{align*}
        \underset{n\to \infty}{\liminf} \frac{1}{r_n}\sum_{i = 1}^n\sum_{j\neq i}(P^{\gamma_{n,\eta}}_{ij})^2 > 0.
    \end{align*}
\end{proposition}
Note that Assumption~\ref{Assumption-High-Level} is satisfied in particular when the entries of~$Z$ are i.i.d~Gaussian. The latter was used in \cite{Hansen:2014ie} to provide a set of sufficient conditions for their assumptions to be satisfied. The proof of Proposition~\ref{prop:hd} also sheds further light on the exact value of~$\eta$. To summarize, Propositions~\ref{Proposition-Weaker-CMSMS} and~\ref{prop:hd} show that Assumption \ref{Assumption-High-Level} is likely to be satisfied; this can be the case even in settings in which the assumptions underlying previous methods are violated. 


 If very small values of $\frac{1}{r_n}\sum_{i = 1}^n\sum_{j\neq i}(P^{\gamma_n}_{ij})^2$ are observed, then Assumption \ref{Assumption-High-Level} may be questionable. If this is the case, practitioners could make Assumption \ref{Assumption-High-Level} more plausible by suitably modifying the matrix of IVs, $Z$. This could include the dropping of IVs (either randomly or motivated by economic reasoning, but not informed by the in-sample correlation of the IVs with the endogenous variable), or the (iterative) removal of certain observations. If sufficient data is available to split the sample, a data-driven selection of IVs on one split of the sample for use on the other is also possible.

\subsubsection{Asymptotic normality of the RJAR test statistic}
We are now in a position to state the asymptotic distribution of the RJAR test under the null hypothesis.
\begin{theorem}\label{Theorem-RJAR}
Suppose Assumptions \ref{Assumption-MC}--\ref{Assumption-High-Level} and the null hypothesis in Equation \eqref{Equation-RJAR-Null-Hypothesis} hold. Consider any sequence $\gamma_n$ such that $\underset{n\to \infty}{\liminf} \frac{1}{r_n}\sum_{i = 1}^n\sum_{j\neq i}(P^{\gamma_n}_{ij})^2 > 0$. Then, the statistic $RJAR_{\gamma_n}(\beta_0)$ defined in Equation \eqref{eq: RJAR} satisfies%
	\begin{equation*}
		RJAR_{\gamma_n}(\beta_0) \overset{d}{\to}\mathcal{N}[0, 1].
	\end{equation*}
\end{theorem}

\begin{corollary}\label{Proposition:RJAR}
Under Assumptions \ref{Assumption-MC}--\ref{Assumption-High-Level} and the null hypothesis in Equation \eqref{Equation-RJAR-Null-Hypothesis}, the RJAR test given in Definition \ref{def: RJAR} (i.e.~using~$\gamma_n=\gamma^*_n$ for all~$n\in\mathbb{N}$) has asymptotic size $\alpha$.
\end{corollary}

Notice that we did not impose any assumption on the coefficients of the instruments $\Pi$ in the first-stage regression in Equation \eqref{Equation-RJAR-Model}. Therefore, the RJAR test is robust to arbitrarily weak identification.

\subsection{Closest alternatives in the literature}\label{Section-RJAR-Competition}

The RJAR test combines two existing approaches in the literature. It uses ridge-regularisation as in CT to allow for $r_n < k_n$, and jackknifing as in CMS and MS to allow for arbitrary heteroskedasticty in the error terms.

The RJAR test is similar to the ridge-regularised AR test proposed by CT given by
\begin{equation*}
        AR_{CT} = \frac{ne(\beta_0)'P^{\theta}e(\beta_0)}{e(\beta_0)'(I_n - P^{\theta})e(\beta_0)},
\end{equation*}
where $P^{\theta} = Z(Z'Z + \theta I_k)^{-1}$ for some fixed scalar $\theta$ that does not depend on $n$, where $\theta \geq 0$ if $r_n = k_n$, and $\theta > 0$ if $r_n < k_n$. CT show that under the assumption of homoskedastic error terms, $AR_{CT}$ converges to an infinite sum of weighted $\chi^2_1$ distributions. Since the limiting distribution depends on an infinite number of unobserved weights, CT propose a bootstrap procedure to derive the critical values for $AR_{CT}$.
The distinguishing features of the RJAR test are the data-driven and unique penalty parameter used to regularise the projection matrix (we find substantial sensitivity of the performance of $AR_{CT}$ to different values of $\theta$), the robustness to arbitrary heteroskedasticiy in the error terms, and its computational speed (it uses standard Normal asymptotic critical values instead of the bootstrap).

Robustness to arbitrary heteroskedasticity in the error terms is achieved through jackknifing, as in CMS and MS. The distinguishing feature of the RJAR test to the tests proposed in CMS and MS is the use of a ridge-regularised projection matrix $P^{\gamma_n}$, which makes the RJAR test applicable also when $k_n > r_n$, unlike the aforementioned two tests that use the standard least-squares projection matrix $Z'(Z'Z)^{-1}Z'$, and cannot be computed when $k_n > r_n$.

CMS assume $r_n = k_n < n$, and propose the jackknifed AR statistic given by
\begin{equation*}
	AR_{CMS}(\beta_0) := \frac{1}{\sqrt{k_n} \sqrt{\hat{\Phi}_{CMS}(\beta_0)}}\sum_{i = 1}^n \sum_{j \neq i}C_{ij}e_i(\beta_0)e_j(\beta_0),
\end{equation*}
where $C := A - B$, $A := P + \Delta$, $B := (I_n - P)D(I_n - D)^{-1}(I_n - P)$, $\Delta := PD(I_n - D)^{-1}P - \frac 1 2 PD(I_n-D)^{-1} - \frac 1 2 D(I_n-D)^{-1}P$, and $D$ is the diagonal matrix containing the diagonal elements of $P$. $\hat{\Phi}_{CMS}(\beta_0) := \frac{2}{k_n}\sum_{i = 1}^n\*\sum_{j \neq i}C_{ij}^2\*e_i^2(\beta_0)\*e_j^2(\beta_0)$. Under the null hypothesis in Equation \eqref{Equation-RJAR-Null-Hypothesis}, CMS show that $AR_{CMS}(\beta_0)$ converges to a Standard Normal distribution. The CLT underlying this result is a modified version of Lemma A2 of \citet{Chao:2012iz} proposed in \citet[Appendix A.4]{Bekker:2015jm}.

MS also assume that $r_n = k_n < n$, and propose a different jackknifed AR statistic than CMS that can be obtained from $RJAR_{\gamma_n}(\beta_0)$ in Equation \eqref{eq: RJAR} by setting $\gamma_n = 0$, and replacing $\hat{\Phi}_{\gamma_n}(\beta_0)$ with 
\begin{equation}\label{eq: Phi_MS}
	\hat{\Phi}_{MS}(\beta_0) := \frac{2}{k_n}\sum_{i = 1}^n\sum_{j \neq i}\frac{P_{ij}^2}{M_{ii}M_{jj} + M_{ij}^2}\left[e_i(\beta_0)M_ie(\beta_0)\right]\left[e_j(\beta_0)M_je(\beta_0)\right],
\end{equation}
where $M = I_n - P$, and $M_i$ is the $i^{\text{th}}$ row of $M$. The reason why the unregularised jackknifed AR test in MS uses $\hat{\Phi}_{MS}(\beta_0)$ instead of the variance estimator given in Equation \eqref{eq: variance estimator} evaluated at $\gamma_n = 0$, is because, according to their Theorem 4, and the discussion in Section 4.2.1, the former yields higher power than the latter. It follows that the unregularised version of our RJAR test, which arises when $\gamma^*_n = 0$ in Equation \eqref{Equation-Gamma-Star}, will be dominated by the jackknifed AR test of MS in terms of power, and practitioners may prefer the latter when $r_n = k_n$ and the balanced-design assumption is satisfied. This can have implications for the finite-sample performance of the jackknifed AR test of MS compared to the RJAR, which we investigate in Section \ref{Section-RJAR-Simulations}.

The only other existing test that allows for $r_n < k_n$ and arbitrary heteroskedasticty is the Sup Score test of BCCH. BCCH first standardise the IVs as in Equation \eqref{Equation-Standardised}, and then propose the Sup Score statistic given by


\begin{equation*}
	{S}(\beta_0) = \underset{1 \leq j \leq k_n}{\text{ max }} \frac{\left|\frac{1}{\sqrt n}\sum_{i = 1}^ne_i(\beta_0)Z_{ij}\right|}{\sqrt{\frac 1 n \sum_{i = 1}^n\left(e_i(\beta_0)\right)^2Z_{ij}^2}}.
\end{equation*}


BCCH propose to use the critical value $\nu_\alpha = c_{BCCH} \*\mathcal{Q}\left(1-\alpha/(2k_n)\right)$, $c_{BCCH} > 1$ and $\alpha \in (0, 1)$. They show that comparing their Sup Score statistic to this critical value yields a test of the null hypothesis in Equation \eqref{Equation-RJAR-Null-Hypothesis} that has asymptotic size less than or equal to $\alpha$. Being a supremum-norm test suggests that the BCCH Sup Score test will work well with a sparse first stage (i.e. where only a few elements of $\Pi$ are zero), but may have lower power than the RJAR test when the first stage is dense. This is verified in the simulations in Section \ref{Section-RJAR-Simulations}.

\section{Simulations \label{Section-RJAR-Simulations}}

We now investigate the size and power properties of the RJAR test and compare them to those of the tests proposed in CT, CMS, MS, and BCCH. We take our simulation setup from \citet{Hansen:2014ie} who in turn take theirs from BCCH. The DGP is given by
\begin{subequations}
	\label{Equation-RJAR-Simulation-Model}
	\begin{align}
			y_i & = X_i\*\beta + \varepsilon_i \label{Equation-Simulation-Second}\\
	        X_i & = Z_{i}'\*\pi + v_i, \label{Equation-Simulation-First}
	\end{align}
\end{subequations}
for $i = 1, \dots, n = 100$ . The IVs $Z_i$ are independent and identically Gaussian with mean 0 and $\text{Var}\left[Z_{il} \right] = 0.3 $ and $\text{Corr}\left[ Z_{il}, Z_{im}\right] = 0.5^{|l-m|}$.\footnote{Simulations in Appendix \ref{Appendix-Simulations-Uncorrelated-IVs} show that the results remain qualitatively unchanged if uncorrelated Gaussian IVs are considered instead.} The error terms are given by
\begin{equation*}
    \begin{aligned}
\begin{bmatrix}
		\varepsilon_{i}\\
		v_{i}
	\end{bmatrix} &\sim \mathcal N \left[0,   
\left[\begin{smallmatrix}
\sigma_\varepsilon^2 & \sigma_{\varepsilon\*v} \\ 
 \sigma_{\varepsilon\*v} &\sigma_v^2 
\end{smallmatrix}\right]
\right]
\end{aligned}
\end{equation*}
with $\sigma_\varepsilon^2=2$, $\sigma_v^2=1$ and $\sigma_{\varepsilon v} = 0.6\sigma_\varepsilon\sigma_v$.

$\pi = \zeta\*\kappa$, where $\kappa$ is a vector of zeros and ones that varies with the type of DGP considered (sparse or dense, as modelled below), and $\zeta$ is some scalar that ensures that for a given concentration parameter, $\mu^2$, the following relationship is satisfied:
\begin{equation*}
	\mu^2 = \frac{n\*\pi'\mathbb{E}\left[Z_{i}\*Z_{i}'\right]\*\pi}{\sigma_v^2}.
\end{equation*}

This implies that
\begin{equation*}
	\zeta = \sqrt{\frac{\sigma_v^2\*\mu^2}{n\kappa'\mathbb{E}\left[Z_{i}\*Z_{i}'\right]\kappa}}.
\end{equation*}

To illustrate how the sparsity structure of the first stage in Equation \eqref{Equation-Simulation-First} can affect the size and power of the studied tests, we consider both a sparse first stage and a dense first stage. Sparsity in the first stage is modelled by setting $\kappa = [\iota_5', 0_{k_n-5}' ]'$, where $\iota_q$ is a $q\times 1$ vector of ones, and $0_q$ is a $q\times 1$ vector of zeros. Denseness in the first stage is modelled by setting $\kappa = [\iota_{0.4k_n}', 0_{0.6k_n}']'$.  

We consider $k_n = 30, 90, 190$. In the context of Assumption \ref{Assumption-High-Level}, we search over values greater than 1 when choosing $\gamma^*_n$ in case $r_n<k_n$. For the case of 30 IVs, the RJAR test does not impose any regularisation ($\gamma^*_n = 0$). For the case of 90 IVs, $\gamma^*_n = 12.048$. We note that in the latter case~$\sum_{i = 1}^n\sum_{j\neq i}(P^{\gamma^*_n}_{ij})^2 = 11.123 > 8.845  = \sum_{i = 1}^n\sum_{j\neq i}(P_{ij})^2$. This shows that even in the case where $r_n < n$, ridge regularisation can make Assumption \ref{Assumption-High-Level} strictly more plausible. For the case of 190 IVs, $\gamma^*_n = 109.187$. 

The variance estimator of MS occasionally yields a negative value. These cases are conservatively interpreted as a failure to reject the null hypothesis. As recommended by BCCH, $c_{BCCH} = 1.1$. As in the simulation section in CT, we set $\theta = 0.05$. The number of Monte Carlo replications is 10,000.

 Appendix \ref{Appendix-Hetero} provides additional simulation evidence on inference with heteroskedastic error terms. Appendix \ref{Appendix-Projection} provides additional simulation evidence on projection-based inference.

\subsection{Size}

Figure \ref{Figure-Size-Homo} shows the simulation results with a sparse first stage for tests of size 0.01 to 0.99, that is the rejection frequency under $H_0:\beta_0 = 1$. As far as the illustration of the tests' size properties is concerned, the dense first stage yields virtually the same results, and is hence omitted. Since all tests are robust to weak IVs, the rejection frequencies of the tests are not affected by the strength of identification. For the case of 30 IVs, the AR tests of CT, CMS, and MS and the RJAR test have correct size, while the BCCH Sup Score test is undersized. 

For the case of 90 IVs, the AR test of CT and the RJAR test control size. The AR test of CMS appears to control size for common small nominal test sizes (e.g., 0.05 or 0.1). The AR test of MS appears to be generally oversized. For example, at nominal level 0.05, the rejection frequency of the test is 0.189.\footnote{The simulations in Appendix \ref{Appendix-Simulations-1000} show that if $n$ and $k_n$ are increased to about 1,000 and 900, respectively, there is virtually no more size distortion.} The BCCH Sup Score test continues to be undersized. 

For the case of 190 IVs, only the AR test of CT, the BCCH Sup Score test, and the RJAR test are feasible. As before, the BCCH Sup Score test is undersized, while the AR test of CT and the RJAR test have correct size. 

\begin{figure}[H]
    \centering
    \begin{subfigure}[b]{0.5\textwidth}
        \includegraphics[width=\textwidth]{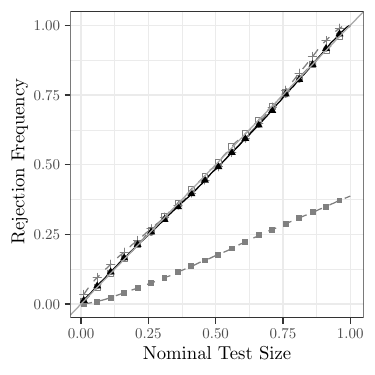}
        \caption{$k_n = 30$}
    \end{subfigure}%
    \begin{subfigure}[b]{0.5\textwidth}
        \includegraphics[width=\textwidth]{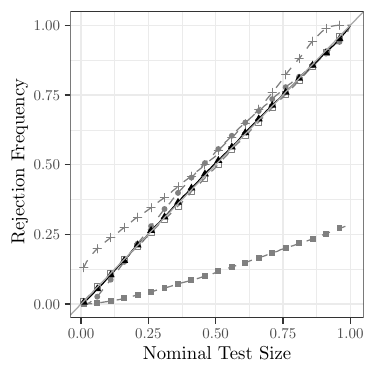}
        \caption{$k_n = 90$}
    \end{subfigure}
  \begin{subfigure}[b]{0.5\textwidth}
        \includegraphics[width=\textwidth]{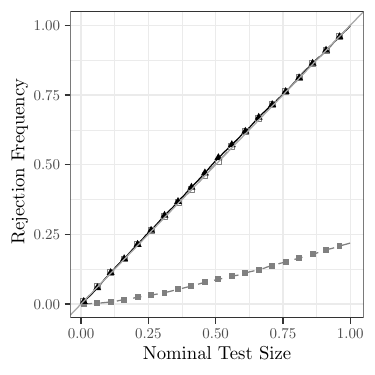}
        \caption{$k_n = 190$}
    \end{subfigure}
    
    \begin{subfigure}[c]{\textwidth}
    	\centering
    	\includegraphics{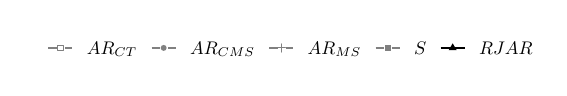}
    \end{subfigure}
    \caption{PP Plots for Sparse IVs, homoskedastic errors, $\beta = 1$, $\mu^2 = 0$, $H_0: \beta_0 = 1$.}
    \label{Figure-Size-Homo}
\end{figure}

\subsection{Power}

Figures \ref{Figure-Power-30-Homo}--\ref{Figure-Power-190-Homo} show the power of the tests when the number of IVs and the sparsity pattern of the first stage is varied. It is still the case that $H_0:\beta_0 = 1$.  

For the case of 30 IVs (Figure \ref{Figure-Power-30-Homo}), the AR tests of CT and CMS have similar power to the RJAR test, while MS is slightly more powerful than the RJAR test. The BCCH Sup Score test is less powerful than all other tests.

For the case of 90 sparse IVs (Figure \ref{Figure-Power-90-Homo}), the RJAR test is slightly more powerful than the BCCH Sup Score test. The AR test of MS fails to control the size, while the AR tests of CT and CMS exhibit power properties substantially worse than those of the BCCH Sup Score test and the RJAR test. For the case of 90 dense IVs (Figure \ref{Figure-Power-90-Homo}), the RJAR test is substantially more powerful than all other alternatives. For the case of 190 sparse and dense IVs (Figure \ref{Figure-Power-190-Homo}), the RJAR test is more powerful than the BCCH Sup Score test.\footnote{The poor power of the BCCH Sup Score is likely a consequence of the Union Bound underlying their Lemma 5. We note that the poor power properties of the BCCH Sup Score are acknowledged in BCCH, and are what leads them to recommend using their main estimator (Post-LASSO two-stage least squares) whenever weak identification is not a concern, and the first stage is sufficiently sparse.} The RJAR test is also more powerful than the AR test of CT. Thus, for all the DGPs that are considered here, the RJAR test is as powerful as existing methods whenever these are applicable, and sometimes much more powerful.

\begin{figure}[H]
    \begin{subfigure}[b]{0.5\textwidth}
        \includegraphics[width=\textwidth]{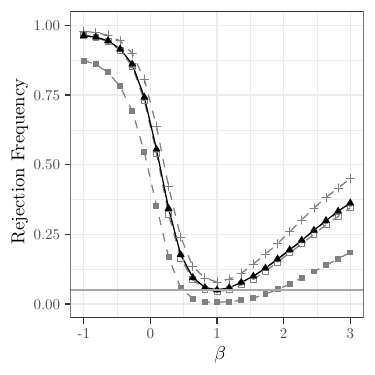}
        \caption{$\mu^2 = 30$, sparse IVs}
    \end{subfigure}%
    \begin{subfigure}[b]{0.5\textwidth}
        \includegraphics[width=\textwidth]{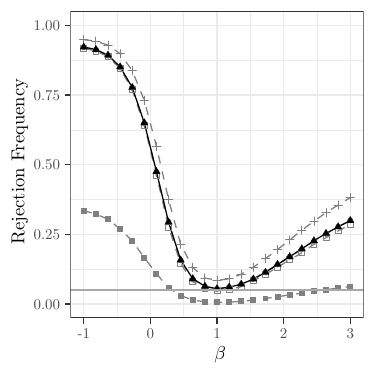}
        \caption{$\mu^2 = 30$, dense IVs}
    \end{subfigure}
    \begin{subfigure}[b]{0.5\textwidth}
        \includegraphics[width=\textwidth]{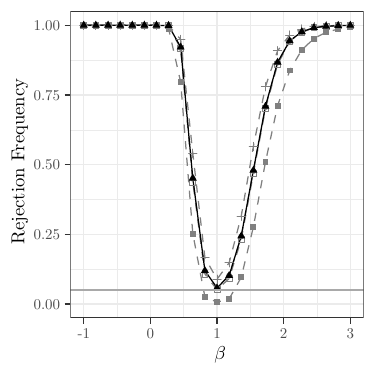}
        \caption{$\mu^2 = 180$, sparse IVs}
    \end{subfigure}%
    \begin{subfigure}[b]{0.5\textwidth}
        \includegraphics[width=\textwidth]{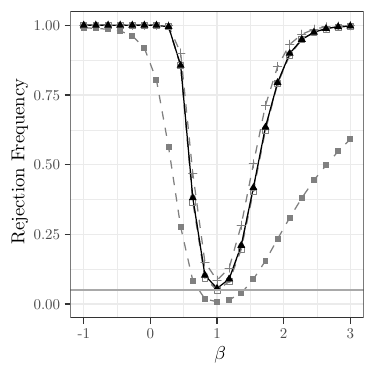}
        \caption{$\mu^2 = 180$, dense IVs}
    \end{subfigure}
    \begin{subfigure}[c]{\textwidth}
    	\centering
    	\includegraphics{TablesFigures/PowerRevision/legend_ld.pdf}
    \end{subfigure}
    \caption{Power curves for 30 IVs. Nominal test size of 5\% indicated by the grey horizontal line. $H_0: \beta_0 = 1$.}
    \label{Figure-Power-30-Homo}
\end{figure}

\begin{figure}[H]
    \begin{subfigure}[b]{0.5\textwidth}
        \includegraphics[width=\textwidth]{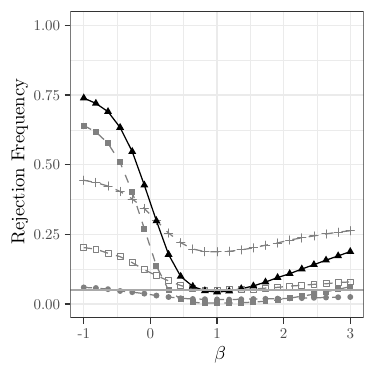}
        \caption{$\mu^2 = 30$, sparse IVs}
    \end{subfigure}%
    \begin{subfigure}[b]{0.5\textwidth}
        \includegraphics[width=\textwidth]{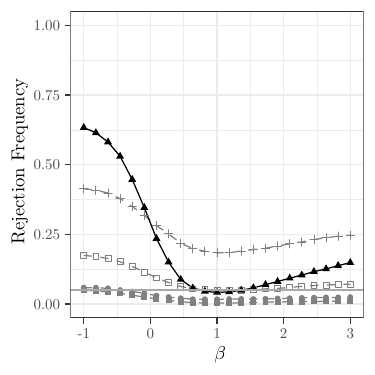}
        \caption{$\mu^2 = 30$, dense IVs}
    \end{subfigure}
    \begin{subfigure}[b]{0.5\textwidth}
        \includegraphics[width=\textwidth]{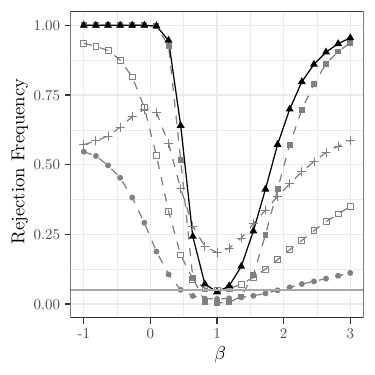}
        \caption{$\mu^2 = 180$, sparse IVs}
    \end{subfigure}%
    \begin{subfigure}[b]{0.5\textwidth}
        \includegraphics[width=\textwidth]{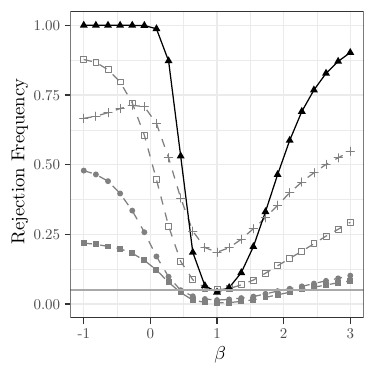}
        \caption{$\mu^2 = 180$, dense IVs}
    \end{subfigure}
    \begin{subfigure}[c]{\textwidth}
    	\centering
    	\includegraphics{TablesFigures/PowerRevision/legend_ld.pdf}
    \end{subfigure}
    \caption{Power curves for 90 IVs. Nominal test size of 5\% indicated by the grey horizontal line. $H_0: \beta_0 = 1$. }
    \label{Figure-Power-90-Homo}
\end{figure}

\begin{figure}[H]
    \begin{subfigure}[b]{0.5\textwidth}
        \includegraphics[width=\textwidth]{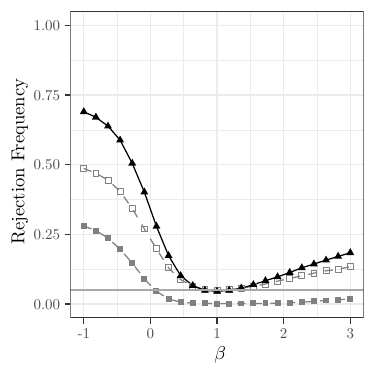}
        \caption{$\mu^2 = 30$, sparse IVs}
    \end{subfigure}%
    \begin{subfigure}[b]{0.5\textwidth}
        \includegraphics[width=\textwidth]{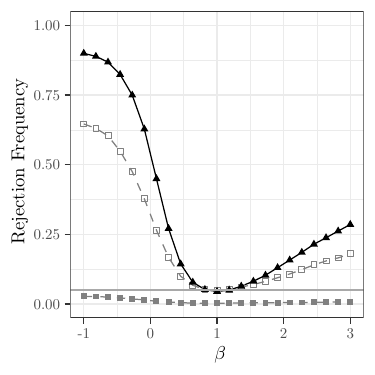}
        \caption{$\mu^2 = 30$, dense IVs}
    \end{subfigure}
    \begin{subfigure}[b]{0.5\textwidth}
        \includegraphics[width=\textwidth]{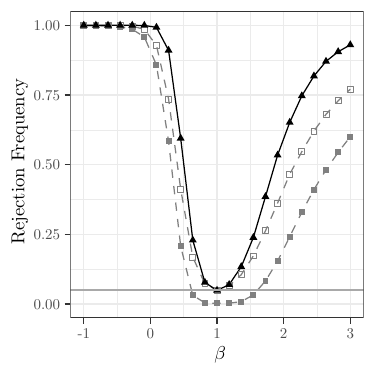}
        \caption{$\mu^2 = 180$, sparse IVs}
    \end{subfigure}%
    \begin{subfigure}[b]{0.5\textwidth}
        \includegraphics[width=\textwidth]{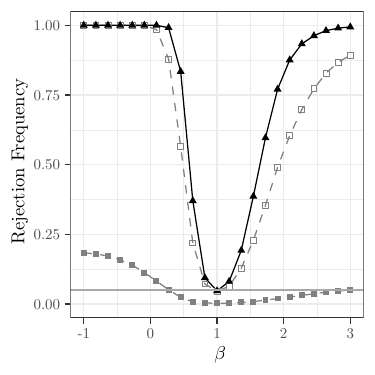}
        \caption{$\mu^2 = 180$, dense IVs}
    \end{subfigure}
    \begin{subfigure}[c]{\textwidth}
    	\centering
    	\includegraphics{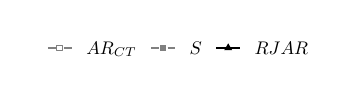}
    \end{subfigure}
    \caption{Power curves for 190 IVs. Nominal test size of 5\% indicated by the grey horizontal line. $H_0: \beta_0 = 1$.}
    \label{Figure-Power-190-Homo}
\end{figure}

\section{Empirical Application \label{Section-RJAR-App}}

We consider an empirical application based on \citet{Card:2009hj}. The coefficient of interest is given by $\beta_s$ in the following model:
\begin{equation}
	\label{Equation-RJAR-Card-2}
	\begin{aligned}
		y_{is} &= \beta_s X_{is} + \delta_s'W_i + \varepsilon_{is}, \\
	\end{aligned}
\end{equation}
where $y_{is}$ is the difference between residual log wages for immigrant and native men in skill group $s$ in city $i$,\footnote{As discussed in \citet[p.~11, footnote 17]{Card:2009hj}, residual wages are wages once observed characteristics of the entire US workforce are controlled for.} $X_{is}$ is the log ratio of immigrant to native hours worked in skill group $s$ of both men and women in city $i$, and $W_i$ is a vector of city-level controls with coefficient vector $\delta_s$, and $\varepsilon_{is}$ is the structural error. In the context of the production function specified in \citet[Section I]{Card:2009hj}, $\beta_s$ can be interpreted as the (negative) inverse elasticity of substitution  between immigrants and natives in the US in their respective skill group. As in \citet{Card:2009hj}, we consider two skill groups $s = h, c$ (high school or college equivalent) separately. 

\citet{Card:2009hj} raises the concern that unobserved factors in a city may lead to both higher wages and higher employment levels of immigrants relative to natives, causing $X_{is}$ to be endogenous. \citet{Card:2009hj} proposes to use the ratio of the number of immigrants from country $l$ in city $i$ to the total number of immigrants from foreign country $l$ in the US as an IV. The rationale for these IVs is that existing immigrant enclaves are likely to attract additional immigrant labour through social and cultural channels unrelated to labour market outcomes. We consider two sets of IVs. First, we consider the original setup of \citet{Card:2009hj}, using as IVs the $k_n = 38$ different countries of origin of the immigrants. Second, motivated by the saturation approach of \citet{Blandhol2022}, we consider the setup where these 38 original IVs are interacted with the $q = 9$ available controls (including a constant). This yields $k_n = 342$ IVs. In both cases, the number of observations (i.e., the number of cities) is $n = 124$.\footnote{We note that the controls and the IVs are at the city level, and hence the same for the applications to high-school workers and college workers. Therefore, the projection matrix and the ridge-regularised projection matrix will also be the same across these two skill groups.}

We construct (weak-identification robust) confidence sets for $\beta_s$ by inverting the AR tests of CT, CMS and MS, the Sup Score test, and the RJAR test. Thus, the 95\% confidence set for any test is obtained as the collection of $\beta_{s,0}$ for which that test does not reject the null at 5\% level of significance. As in the simulation exercise in Section \ref{Section-RJAR-Simulations}, we search over values greater than 1 when choosing~$\gamma^*_n$ in case~$r_n<k_n$. Again $c_{BCCH} = 1.1$ and $\theta = 0.05$. The number of boostrap replications for the AR test of CT is set to 2,500. A grid of 100 values for $\beta_{s,0}$ is used for $s = h, c$. Data is taken from a single cross section, as made available by \citet{GoldsmithPinkham:2020hm}.

Figure \ref{Figure-Card-38} shows the confidence sets when $k_n = 38$ for high-school workers and college workers. We find that $\gamma^*_n = 0$, implying that no regularisation is needed. This is in line with our simulations in Section \ref{Section-RJAR-Simulations}, where we found that regularisation was not needed to maximise the sum in Assumption \ref{Assumption-High-Level} when $k_n/n = 0.3$. Furthermore, $\underset{1 \leq i \leq n}{\text{ max }}P_{ii} = 0.944$, and  $r^{-1}_n\sum_{i = 1}^n\sum_{j\neq i}(P^{\gamma^*_n}_{ij})^2 = 0.513$. 
We also point out that only three diagonal entries of $P$ are larger than $0.9$, which suggests that Assumption \ref{Assumption-High-Level} is reasonably satisfied by part 2.~of Proposition \ref{Proposition-Weaker-CMSMS}. The 95\% confidence sets for each test are given by all the points below the grey horizontal line. The confidence sets for both skill groups broadly confirm the results in \citet{Card:2009hj}. We find that the confidence sets for high-school workers is smallest for the jackknifed AR statistic of CMS, whereas the BCCH Sup Score test yields the smallest confidence interval for the application to college workers. For both cases, the AR test of CT yields the largest confidence interval. Based on the power results in Section \ref{Section-RJAR-Simulations}, this suggests a very sparse first stage for college workers, i.e., a few nationalities being highly predictive of inflows of immigrant labour.


\begin{figure}[H]
    \centering
    \begin{subfigure}[c]{0.5\textwidth}
        \centering
        \includegraphics[width = 0.9\textwidth]{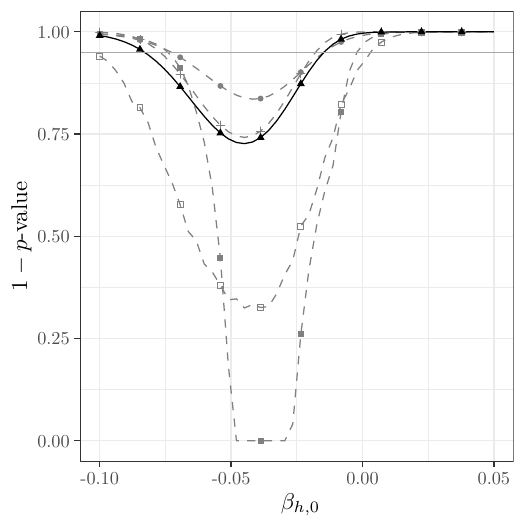}
        \caption{High-school workers}
    \end{subfigure}%
    \begin{subfigure}[c]{0.5\textwidth}
        \centering
        \includegraphics[width = 0.9\textwidth]{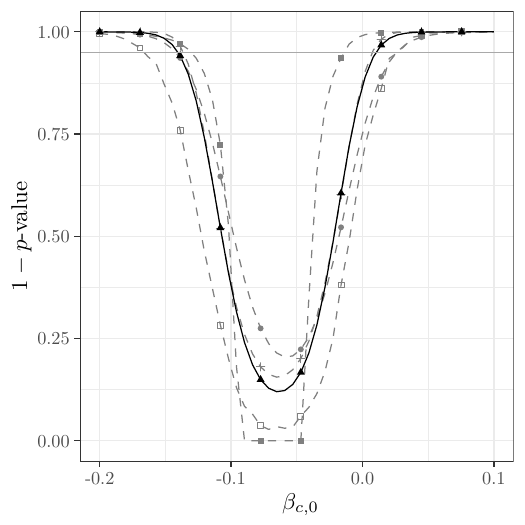}
        \caption{College workers}
    \end{subfigure}
    \begin{subfigure}[c]{\textwidth}
    	\centering
    	\includegraphics{TablesFigures/PowerRevision/legend_ld.pdf}
    \end{subfigure}
    \caption{95 \% confidence sets for $\beta_s$ for the application in Equation \eqref{Equation-RJAR-Card-2} with $k_n = 38$ IVs. $\text{max}_iP_{ii} = 0.944$. $\gamma_n^* = 0$, $r^{-1}_n\sum_{i = 1}^n\sum_{j\neq i}(P^{\gamma^*_n}_{ij})^2 = 0.513$.}
    \label{Figure-Card-38}
\end{figure}

Figure \ref{Figure-Card-342} shows the confidence sets when $k_n = 342$ for high-school workers and college workers, respectively. Since $r_n = n - q = 124-9 = 115 < 342 = k_n$, the jackknifed AR statistics of CMS and MS are not applicable. We find that $\gamma^*_n = 5.299$ and $r^{-1}_n\sum_{i = 1}^n\sum_{j\neq i}(P^{\gamma^*_n}_{ij})^2 = 0.106$. The 95\% confidence sets obtained by inverting the AR test of CT are empty for both $\beta_{h}$ and $\beta_{c}$. This could be due to heteroskedasticity in the error terms. In line with the simulation results on power reported in Section \ref{Section-RJAR-Simulations}, the RJAR test yields smaller confidence intervals than the BCCH Sup Score test. The qualitative conclusions with respect to the case of 38 IVs remain unchanged. 

\begin{figure}[H]
    \centering
    \begin{subfigure}[c]{0.5\textwidth}
        \centering
        \includegraphics[width = 0.9\textwidth]{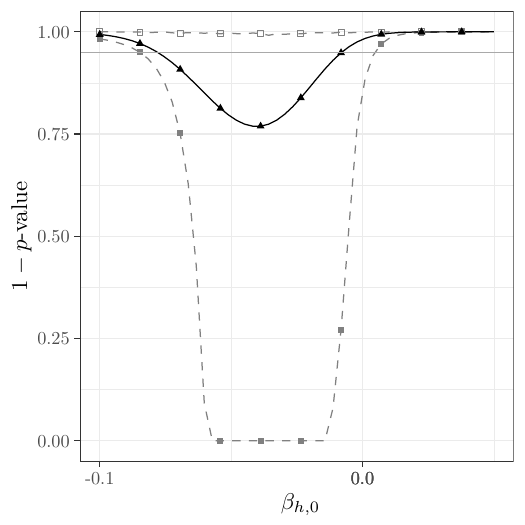}
        \caption{High-school workers}
    \end{subfigure}%
    \begin{subfigure}[c]{0.5\textwidth}
        \centering
        \includegraphics[width = 0.9\textwidth]{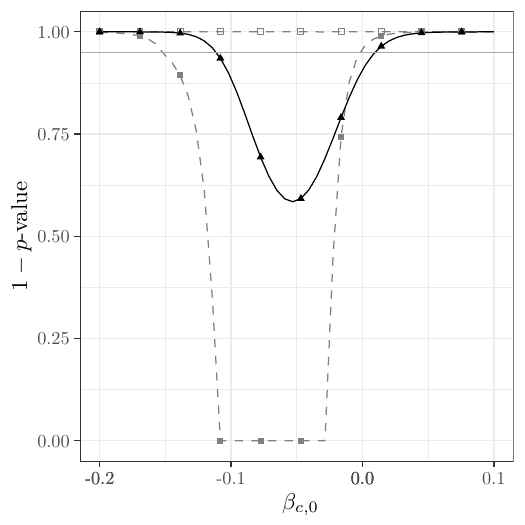}
        \caption{College workers}
    \end{subfigure}
    \begin{subfigure}[c]{\textwidth}
    	\centering
    	\includegraphics{TablesFigures/PowerRevision/legend_hd.pdf}
    \end{subfigure}
    \caption{95 \% confidence sets for $\beta_s$ for the application in Equation \eqref{Equation-RJAR-Card-2} with $k_n = 342$ IVs. $\gamma_n^* = 5.299$, $r^{-1}_n\sum_{i = 1}^n\sum_{j\neq i}(P^{\gamma^*_n}_{ij})^2 = 0.106$.}
    \label{Figure-Card-342}
\end{figure}

\section{Conclusion \label{Section-RJAR-Conclusion}}

We contributed to the literature on (very) many IVs in the cross-sectional linear IV model by proposing a new, ridge-regularised jackknifed AR test. Our test compares favourably with existing methods in the literature both theoretically, by allowing for high-dimensional IVs and weakening a common assumption on the IVs' projection matrix, and practically, by having correct asymptotic size and displaying favourable power properties even when the number of IVs approaches or exceeds the number of observations.

\pagebreak

\bibliographystyle{chicago}

\bibliography{RJAR}


\renewcommand{\thesection}{\Alph{section}}
\numberwithin{equation}{section}
\newpage
\setcounter{section}{0}

\bigskip
\begin{center}
{\large\bf SUPPLEMENTARY MATERIAL}
\end{center}

\section{Proofs}

Throughout this appendix, $C>0$ denotes a universal constant that can change across lines. CSHNW refers to \citet{Chao:2012iz} from where we also borrow our summation conventions. We also define $\Phi_{\gamma_n} :=  \frac{2}{r_n}\sum_{i = 1}^n\sum_{j \neq i}(P_{ij}^{\gamma_n})^2\mathbb{E}[\varepsilon_i^2]\mathbb{E}[\varepsilon_j^2]$. Since all the proofs in this appendix are under the null hypothesis, we write $\varepsilon_i$ and $\hat{\Phi}_{\gamma_n}$ for $e_i(\beta_0)$ and $\hat{\Phi}_{\gamma_n}(\beta_0)$, respectively. 

Furthermore, we note that Assumption \ref{Assumption-High-Level} implies that 
\begin{equation}
    \label{Equation-Assumption-High-Level-Implication}
\sum_{i = 1}^n\sum_{j\neq i}(P_{ij}^{\gamma_n})^2 \geq cr_n,
\end{equation}
for some $c > 0$ and~$n$ sufficiently large.

Our strategy for deriving the limiting distribution of our RJAR test is as follows. Since $P^{\gamma_n}$ is generally not idempotent for $\gamma_n \geq 0$, we cannot rely on standard properties of idempotent matrices. Therefore, we first derive a series of properties for the ridge-regularised projection matrix $P^{\gamma_n}$ that hold for any $\gamma_n\geq 0$ when $r_n = k_n$ and $\gamma_n >0$ when $r_n < k_n$. These properties are derived using a singular value decomposition of the matrix of IVs $Z$, and are collected in Lemmas \ref{Lemma-P-Linear-Algebra} and \ref{Lemma-B2}. These Lemmas play a similar role to Lemma B1 and Lemma B2 in \citet{Chao:2014}. The proof of Theorem \ref{Theorem-RJAR} derives the asymptotic distribution of our RJAR statistic under the assumption that $\Phi_{\gamma_n}$ is known. This is achieved by casting the RJAR statistic as a degenerate U statistic with variable kernel, and verifying the conditions of a Martingale CLT. A benefit of deriving our proofs from the bottom up (and not relying on the augmentation approach in \citet{Hansen:2014ie}) is that we do not require that all diagonal elements of $P^{\gamma_n}$ be bounded away from 1. Lemma \ref{Lemma-Variance-Cons} shows that our proposed estimator of $\Phi_{\gamma_n}$ is consistent. Taken together with Theorem \ref{Theorem-RJAR} and Assumption \ref{Assumption-High-Level}, Lemma \ref{Lemma-Variance-Cons} implies that our RJAR test is asymptotically valid. Finally, Lemma \ref{Lemma-Gamma-Star-Existence} proves the existence of $\gamma^*_n$, and the proof of Proposition \ref{Proposition-Weaker-CMSMS} shows that Assumption 3 is weaker than the balanced-design assumption in CMS and MS.

The following singular value decomposition of the $n\times k_n$ matrix of instruments $Z$ of rank $r_n$ will be used frequently (see, e.g., \citet[p. 60]{Lutkepohl:1996uz}):
\begin{equation*}
	Z = USQ',
\end{equation*}
where $U$ is an $n\times n$ matrix such that $U'U=UU'=I_n$, $Q$ is a $k_n\times k_n$ matrix such that $Q'Q = QQ' = I_{k_n}$, and $S$ is the $n\times k_n$ matrix given by
\begin{equation*}
	S = \begin{bmatrix}
		D & 0_{r_n \times (k_n - r_n)}\\
		0_{(n - r_n)\times r_n} & 0_{(n - r_n) \times (k-r_n)}\\
	\end{bmatrix},
\end{equation*}
where $D$ is the diagonal $r_n\times r_n$ matrix containing the singular values of $Z$. Hence, one can write
\begin{equation}\label{eq:Pgamma}
	\begin{aligned}
	P^{\gamma_n} = USQ'(QS'SQ' + Q\gamma_n I_{k_n}Q')^{-1}QS'U'
			 = US(S'S + \gamma_n I_{k_n})^{-1}S'U'
			 = U\tilde{D}U',
	\end{aligned}
\end{equation}
where $\tilde D =S(S'S + \gamma_n I_{k_n})^{-1}S'$ is the diagonal $n\times n$ matrix with diagonal entries given by given by $\tilde{D}_{ll} = \frac{D_{ll}^2}{D_{ll}^2 + \gamma_n }\leq 1$ for $l = 1, \dots, r_n$, and zero otherwise. Note that the diagonal entries of $\tilde D$ are also the the eigenvalues of $P^{\gamma_n}$. 

\subsection{Lemmas}
The following lemma collects some properties of $P^{\gamma_n}$. Recall that $Z$ has rank $r_n$.
\begin{lemma}
	Fix $n\geq 3$. For all $h, m = 1, \dots, n$ and $\gamma_n  \geq 0$ if $r_n = k_n$ and $\gamma_n > 0$ if $r_n < k_n$ one has \label{Lemma-P-Linear-Algebra} 
	\begin{enumerate}[label = (\roman*)]
		\item $0\leq (P^{\gamma_n})^j_{hh} \leq P^{\gamma_n}_{hh} \leq 1$ for all positive integers $j$, \label{Lemma-P-Squared-Ineq}  
		\item $\sum_{i = 1}^n(P^{\gamma_n}_{hi})^2 = (P^{\gamma_n})^2_{hh} \leq P^{\gamma_n}_{hh}$, \label{Lemma-P-Row-Ineq}
		\item $\sum_{i = 1}^nP^{\gamma_n}_{ii} = \sum_{l = 1}^{r_n}\frac{D_{ll}^2}{D_{ll}^2 + \gamma_n } \leq r_n$, \label{Lemma-Traces-1}
		\item $\sum_{i = 1}^n(P^{\gamma_n})^2_{ii} = \sum_{l = 1}^{r_n}\frac{D_{ll}^4}{\left(D_{ll}^2 + \gamma_n \right)^2}\leq r_n$,\label{Lemma-Traces-2}
		\item \label{Lemma-One-Bounds}
 $|P^{\gamma_n}_{hm}| \leq 1$, \label{Lemma-Off-Diag-One}
		\item 	 For any $\mathcal{I}_2\subseteq \{1,\hdots,n\}^2$ and any $\mathcal{I}_3\subseteq\{1,\hdots,n\}^3$, \label{Lemma-B1} 
	\begin{enumerate}
		\item $\sum_{\mathcal{I}_2}(P^{\gamma_n}_{ij})^4 \leq r_n$,
		\item $\sum_{\mathcal{I}_3}(P^{\gamma_n}_{ij})^2(P^{\gamma_n}_{jl})^2 \leq r_n$.
	\end{enumerate}
	\end{enumerate}

\end{lemma}

\begin{proof}
By the arguments prior to this lemma, the non-zero eigenvalues of $P^{\gamma_n}$ are $\tilde{D}_{ll}=\frac{D_{ll}^2}{D_{ll}^2 + \gamma_n}$ for $l = 1, \dots, r_n$. \ref{Lemma-Traces-1} now follows from the trace of $P^{\gamma_n}$ equaling the sum of its non-zero eigenvalues. Furthermore, the eigenvalues of $(P^{\gamma_n})^2$ are given by $\tilde{D}_{ll}^2 = \frac{D_{ll}^4}{(D_{ll}^2 + \gamma_n )^2}$ for $l = 1, \dots, r_n$, and zero otherwise. \ref{Lemma-Traces-2} now follows in the same way as \ref{Lemma-Traces-1}.

By Equation \eqref{eq:Pgamma} one has that
\begin{equation}
	\label{Equation-P-hh}
(P^{\gamma_n})^j_{hh} = \sum_{l = 1}^n\tilde{D}_{ll}^jU_{hl}^2 \leq \sum_{l = 1}^n\tilde{D}_{ll}U_{hl}^2 = P^{\gamma_n}_{hh} \leq 1,
\end{equation}
which verifies \ref{Lemma-P-Squared-Ineq} since $0 \leq \tilde{D}_{ll} \leq 1$; the last inequality following from the largest diagonal entry of $P^{\gamma_n}$ being bounded from above by its largest eigenvalue, which is no greater than one. Next, by \ref{Lemma-P-Squared-Ineq},
\begin{align*}
    \sum_{i = 1}^n(P^{\gamma_n}_{hi})^2 = (P^{\gamma_n})^2_{hh}\leq P^{\gamma_n}_{hh}\leq 1,
\end{align*}
such that \ref{Lemma-P-Row-Ineq} and \ref{Lemma-Off-Diag-One} follow. Furthermore, by \ref{Lemma-Off-Diag-One} and \ref{Lemma-Traces-2},
	\begin{equation*}
		\sum_{\mathcal{I}_2}(P^{\gamma_n}_{ij})^4 \leq \sum_{i = 1}^n\sum_{j = 1}^n(P^{\gamma_n}_{ij})^2 = \sum_{i = 1}^n(P^{\gamma_n})^2_{ii} \leq r_n,
	\end{equation*}
	and \ref{Lemma-B1} follows since
	\begin{equation*}
		\sum_{\mathcal{I}_3}(P^{\gamma_n}_{ij})^2(P^{\gamma_n}_{jl})^2  \leq \sum_{j = 1}^n\sum_{i = 1}^n(P^{\gamma_n}_{ij})^2\sum_{l = 1}^n(P^{\gamma_n}_{jl})^2 = \sum_{j = 1}^n\left((P^{\gamma_n})^2_{jj}\right)^2 \leq \sum_{j = 1}^n(P^{\gamma_n})^2_{jj} \leq r_n,
	\end{equation*}
the penultimate inequality being a consequence of \ref{Lemma-P-Squared-Ineq} and the last of \ref{Lemma-Traces-2}.

\end{proof}

\begin{lemma}
\label{Lemma-B2}
	Fix $n\geq 4$. For all $\gamma_n  \geq 0$ if $r_n = k_n$ and $\gamma_n > 0$ if $r_n < k_n$ one has
	\begin{equation*}
		\left|\sum_{i < j < l < m}P^{\gamma_n}_{il}P^{\gamma_n}_{jl}P^{\gamma_n}_{im}P^{\gamma_n}_{jm}\right| \leq Cr_n.
	\end{equation*}
\end{lemma}	 
	
	\begin{proof}
	The proof follows Lemma B2 in CSHNW closely. We verify that Lemma B2 in CSHNW remains valid for any symmetric matrix satisfying the properties in Lemma \ref{Lemma-P-Linear-Algebra}. For $G := \text{diag}(P^{\gamma_n}_{11}, \dots, P^{\gamma_n}_{nn})$ observe that $(P^{\gamma_n} - G)^4$ equals
	\begin{align*}
		  &(P^{\gamma_n})^4 - (P^{\gamma_n})^3G - (P^{\gamma_n})^2GP^{\gamma_n} + (P^{\gamma_n})^2G^2 - P^{\gamma_n} G (P^{\gamma_n})^2 + P^{\gamma_n} GP^{\gamma_n} G\\
		  + & P^{\gamma_n} G^2P^{\gamma_n} - P^{\gamma_n} G^3 - G(P^{\gamma_n})^3 + G(P^{\gamma_n})^2G + GP^{\gamma_n} GP^{\gamma_n}\\
					 - & GP^{\gamma_n} G^2 + G^2(P^{\gamma_n})^2 - G^2P^{\gamma_n} G - G^3P^{\gamma_n}+ G^4.
	\end{align*}
	Since $\text{tr}(A') = \text{tr}(A)$ and $\text{tr}(AB) = \text{tr}(BA)$ for square matrices $A$ and $B$,
		\begin{align*}
		\text{tr}((P^{\gamma_n} - G)^4) =& \text{tr}((P^{\gamma_n})^4 ) - 4\text{tr}((P^{\gamma_n})^3G) + 4\text{tr}((P^{\gamma_n})^2G^2) + 2\text{tr}(P^{\gamma_n} G P^{\gamma_n} G)\\
			& - 4\text{tr}(P^{\gamma_n} G^3) + \text{tr}(G^4)\\
		\leq & \text{tr}((P^{\gamma_n})^4 ) + 4|\text{tr}((P^{\gamma_n})^3G)| + 4\text{tr}((P^{\gamma_n})^2G^2) + 2\text{tr}(P^{\gamma_n} G P^{\gamma_n} G)\\
			& + 4|\text{tr}(P^{\gamma_n} G^3)| + \text{tr}(G^4).
		\end{align*}

	Next, because $0 \leq (P^{\gamma_n})^j_{ii} \leq P^{\gamma_n}_{ii}$ for all positive integers $j$ by Lemma \ref{Lemma-P-Linear-Algebra} \ref{Lemma-P-Squared-Ineq}, and $G_{ii} \geq 0$, one gets
	\begin{align*}
		\text{tr}((P^{\gamma_n} - G)^4) \leq & \text{tr}(P^{\gamma_n} ) + 4\text{tr}(P^{\gamma_n} G) + 4\text{tr}(P^{\gamma_n} G^2) + 2\text{tr}(P^{\gamma_n} G P^{\gamma_n} G)\\
			& + 4\text{tr}(P^{\gamma_n} G^3) + \text{tr}(G^4).
	\end{align*}
	Since $0 \leq P^{\gamma_n}_{ii} \leq 1$ and $G^j \leq I_n$ (elementwise) for all positive integers $j$, $\text{tr}(P^{\gamma_n} G^j) \leq \text{tr}(P^{\gamma_n}) \leq r_n$, $\text{tr}(P^{\gamma_n} G P^{\gamma_n} G) \leq \text{tr}((P^{\gamma_n})^2) \leq r_n$. Furthermore, $\text{tr}(G^4) \leq \text{tr}(G) = \text{tr}(P^{\gamma_n}) \leq r_n$. Hence, 
	\begin{equation}
		\label{Equation-Inter}
		\text{tr}((P^{\gamma_n} - G)^4) \leq 16r_n.
	\end{equation}
	As in the proof of Lemma B2 in CSHNW, define the lower-triangular matrix $L$ with entries $L_{ij} = P^{\gamma_n}_{ij}\mathbf{1}_{i > j}$, so that $P^{\gamma_n} = L + L' + G$. Then
	\begin{align*}
		(P^{\gamma_n} - G)^4 =& (L + L')^4\\
		            =& L^4 + L^2LL' + L^2L'L + L^2L^{\prime 2} + LL'L^2 + LL'LL' \\
		            & + LL'L'L + LL^{\prime 3} + L'LL^2 + L'LLL' + L'LL'L + L'LL^{\prime 2} \\
					& + L^{\prime 2}L^2 + L^{\prime 2}LL' + L^{\prime 2}L'L + L^{\prime 4}.
	\end{align*}
	Note for all positive integers $j$, $[(L')^j]' = L^j$. Since $\text{tr}(A') = \text{tr}(A)$ and $\text{tr}(AB) = \text{tr}(BA)$ for any square matrices $A$ and $B$,
	\begin{equation}
	    \label{Equation-Expansion-L}
	    \begin{aligned}
		\text{tr}((P^{\gamma_n} - G)^4) = 2 \text{tr}(L^4) + 8\text{tr}(L^3L') + 4\text{tr}(L^2L^{\prime 2}) + 2\text{tr}(L'LL'L).
		\end{aligned}
	\end{equation}	
	We consider each of the terms on the right-hand side of Equation \eqref{Equation-Expansion-L} separately. Before proceeding, note that:
	\begin{align}
	    \label{Equation-L-Expressions}
    	\begin{split}
    	    (L)^2_{ab} &= \sum_{l = 1}^nL_{al}L_{lb},\\
    	    (L')^2_{ab} &= \sum_{l = 1}^nL_{la}L_{bl} = \sum_{l = 1}^nL_{bl}L_{la} = (L)^2_{ba},\\
    	    (L'L)_{ab} &= \sum_{l = 1}^nL_{la}L_{lb},\\
    	    (L)^3_{ab} &= \sum_{l = 1}^n(L)^2_{al}L_{lb} = \sum_{l = 1}^n\sum_{m = 1}^nL_{am}L_{ml}L_{lb},\\
    	    (L)^4_{ab} &= \sum_{l = 1}^n(L)^2_{al}(L)^2_{lb} = \sum_{l = 1}^n\sum_{m = 1}^n\sum_{j = 1}^n L_{am}L_{ml}L_{lj}L_{jb},\\
    	    (L^3L')_{ab} &= \sum_{l = 1}^n(L)^3_{al}(L')_{lb} = \sum_{l = 1}^n\sum_{m = 1}^n\sum_{j = 1}^nL_{am}L_{mj}L_{jl}L_{bl},\\
    	    (L^2L^{\prime 2})_{ab} &= \sum_{l = 1}^n (L)^2_{al}(L')^2_{lb} = \sum_{l = 1}^n (L)^2_{al}(L)^2_{bl} = \sum_{l = 1}^n\sum_{m = 1}^n\sum_{j = 1}^nL_{am}L_{ml}L_{bj}L_{jl},\\
    	    (L'LL'L)_{ab}&= \sum_{l = 1}^n(L'L)_{al}(L'L)_{lb} = \sum_{l = 1}^n\sum_{m = 1}^n\sum_{j = 1}^nL_{ma}L_{ml}L_{jl}L_{jb}.
    	\end{split}
    \end{align}
	
	Using the expression for $(L)^4_{ab}$ in Equation \eqref{Equation-L-Expressions},
	\begin{equation}
	    \label{Equation-Trace-L4}
	    \begin{aligned}
	        \text{tr}(L^4) &= \sum_{i = 1}^n\sum_{l = 1}^n\sum_{m = 1}^n\sum_{j = 1}^n L_{im}L_{ml}L_{lj}L_{ji}= \sum_{i, j, l, m}P^{\gamma_n}_{ji}\mathbf{1}_{j > i}P^{\gamma_n}_{im}\mathbf{1}_{i > m}P^{\gamma_n}_{ml}\mathbf{1}_{m > l}P^{\gamma_n}_{lj}\mathbf{1}_{l > j}\\
	        & = 0,
	    \end{aligned}
	\end{equation}
	since there is no combination of indices $i, j, l, m$ that jointly satisfy each of the indicator functions, as this would require $m < i < j < l < m$.
	
	Using the expression for $(L^3L')_{ab}$ in Equation \eqref{Equation-L-Expressions},
	\begin{equation}
	    \label{Equation-Trace-L3LT}
	    \begin{aligned}
	        \text{tr}(L^3L') &= \sum_{i = 1}^n\sum_{l = 1}^n\sum_{m = 1}^n\sum_{j = 1}^n L_{im}L_{mj}L_{jl}L_{il} = \sum_{i, j, l, m}P^{\gamma_n}_{im}\mathbf{1}_{i > m}P^{\gamma_n}_{mj}\mathbf{1}_{m > j}P^{\gamma_n}_{jl}\mathbf{1}_{j > l}P^{\gamma_n}_{il}\mathbf{1}_{i > l}\\
	        & = \sum_{l < j < m < i} P^{\gamma_n}_{im}P^{\gamma_n}_{mj}P^{\gamma_n}_{jl}P^{\gamma_n}_{il} = \sum_{i < j < l < m}P^{\gamma_n}_{ml}P^{\gamma_n}_{lj}P^{\gamma_n}_{ji}P^{\gamma_n}_{mi}\\
	        &= \sum_{i < j < l < m}P^{\gamma_n}_{ij}P^{\gamma_n}_{jl}P^{\gamma_n}_{im}P^{\gamma_n}_{lm}.
	    \end{aligned}
	\end{equation}
	The third equality follows since the product within the summations is non-zero only for $l < j < m < i$. The fourth equality follows by replacing the indices according to the dictionary $\{i:m, j:j, l:i,  m:l \}$. The last equality follows from the symmetry of $P^{\gamma_n}$.
	
	Using the expression for $(L^2L^{\prime 2})_{ab}$ in Equation \eqref{Equation-L-Expressions},
	\begin{equation}
	    \label{Equation-Trace-L2LT2}
	    \begin{aligned}
	        \text{tr}(L^2L^{\prime 2}) &= \sum_{i = 1}^n \sum_{l = 1}^n\sum_{m = 1}^n\sum_{j = 1}^nL_{im}L_{ml}L_{ij}L_{jl}\\
	        &= \sum_{i, j, l, m} P^{\gamma_n}_{im}\mathbf{1}_{i > m}P^{\gamma_n}_{ml}\mathbf{1}_{m > l} P^{\gamma_n}_{ij}\mathbf{1}_{i > j} P^{\gamma_n}_{jl}\mathbf{1}_{j > l}\\
	        =& \sum_{i > j > l, i > m > l}P^{\gamma_n}_{im}P^{\gamma_n}_{ml}P^{\gamma_n}_{ij}P^{\gamma_n}_{jl} = \sum_{i > j > l, i > m > l}P^{\gamma_n}_{ij}P^{\gamma_n}_{jl}P^{\gamma_n}_{lm}P^{\gamma_n}_{mi}\\
	        =&\sum_{i > j = m > l}P^{\gamma_n}_{ij}P^{\gamma_n}_{jl}P^{\gamma_n}_{lm}P^{\gamma_n}_{mi} + \sum_{i > j > m > l}P^{\gamma_n}_{ij}P^{\gamma_n}_{jl}P^{\gamma_n}_{lm}P^{\gamma_n}_{mi}\\
	        & + \sum_{i > m > j > l}P^{\gamma_n}_{ij}P^{\gamma_n}_{jl}P^{\gamma_n}_{lm}P^{\gamma_n}_{mi}\\
	        =&\sum_{i > j > l}P^{\gamma_n}_{ij}P^{\gamma_n}_{jl}P^{\gamma_n}_{lj}P^{\gamma_n}_{ji} + \sum_{i < j < l < m}P^{\gamma_n}_{ml}P^{\gamma_n}_{li}P^{\gamma_n}_{ij}P^{\gamma_n}_{jm}\\
	        &+ \sum_{i > m > j > l}P^{\gamma_n}_{ij}P^{\gamma_n}_{jl}P^{\gamma_n}_{lm}P^{\gamma_n}_{mi}\\
	        =&\sum_{i > j > l}P^{\gamma_n}_{ij}P^{\gamma_n}_{jl}P^{\gamma_n}_{lj}P^{\gamma_n}_{ji} + \sum_{i < j < l < m}P^{\gamma_n}_{ml}P^{\gamma_n}_{li}P^{\gamma_n}_{ij}P^{\gamma_n}_{jm}\\
	        &+ \sum_{i < j < l < m}P^{\gamma_n}_{mj}P^{\gamma_n}_{ji}P^{\gamma_n}_{il}P^{\gamma_n}_{lm}\\
            =&\sum_{i < j < l}(P^{\gamma_n}_{ij})^2(P^{\gamma_n}_{jl})^2 + 2\sum_{i < j < l < m}P^{\gamma_n}_{ij}P^{\gamma_n}_{il}P^{\gamma_n}_{jm}P^{\gamma_n}_{lm}.\\
	    \end{aligned}
	\end{equation}
	The third equality follows since the product within the summations is non-zero only when both $i > j > l$ and $i > m > l$. The fourth equality follows from the symmetry of $P^{\gamma_n}$. The sixth equality follows by replacing the indices in the second summation according to the dictionary $\{i:m, j:l, l:i, m:j\}$. The seventh equality follows by replacing the indices in the third summation according to the dictionary $\{i:m, j:j, l:i, m:l\}$. The last equality follows from the symmetry of $P^{\gamma_n}$.
	
	Using the expression for $(L'LL'L)_{ab}$ in Equation \eqref{Equation-L-Expressions},
	\begin{align}
	    \label{Equation-Trace-LTLLTL}
	    \begin{split}
	        \text{tr}(L'LL'L) =& \sum_{i = 1}^n \sum_{l = 1}^n\sum_{m = 1}^n\sum_{j = 1}^nL_{mi}L_{ml}L_{jl}L_{ji}\\ 
	        =& \sum_{i, j, l, m} P^{\gamma_n}_{mi}\mathbf{1}_{m > i}P^{\gamma_n}_{ml}\mathbf{1}_{m > l} P^{\gamma_n}_{jl}\mathbf{1}_{j > l} P^{\gamma_n}_{ji}\mathbf{1}_{j > i}\\
	        =& \sum_{i, j, l, m} P^{\gamma_n}_{ij}\mathbf{1}_{i > j}P^{\gamma_n}_{l > j}\mathbf{1}_{l > j} P^{\gamma_n}_{lm}\mathbf{1}_{l > m} P^{\gamma_n}_{im}\mathbf{1}_{i > m}\\
	        =& \sum_{m = j < i = l}P^{\gamma_n}_{ij}P^{\gamma_n}_{lj}P^{\gamma_n}_{lm}P^{\gamma_n}_{im} + \sum_{m = j < l < i}P^{\gamma_n}_{ij}P^{\gamma_n}_{lj}P^{\gamma_n}_{lm}P^{\gamma_n}_{im} \\
	        & + \sum_{m = j < i < l}P^{\gamma_n}_{ij}P^{\gamma_n}_{lj}P^{\gamma_n}_{lm}P^{\gamma_n}_{im}  + \sum_{j < m < i = l}P^{\gamma_n}_{ij}P^{\gamma_n}_{lj}P^{\gamma_n}_{lm}P^{\gamma_n}_{im}\\
	        & + \sum_{m < j < i = l}P^{\gamma_n}_{ij}P^{\gamma_n}_{lj}P^{\gamma_n}_{lm}P^{\gamma_n}_{im}  + \sum_{m < j < l < i}P^{\gamma_n}_{ij}P^{\gamma_n}_{lj}P^{\gamma_n}_{lm}P^{\gamma_n}_{im}\\
	        & + \sum_{m < j < i < l}P^{\gamma_n}_{ij}P^{\gamma_n}_{lj}P^{\gamma_n}_{lm}P^{\gamma_n}_{im} + \sum_{j < m < l < i}P^{\gamma_n}_{ij}P^{\gamma_n}_{lj}P^{\gamma_n}_{lm}P^{\gamma_n}_{im}\\
	        & + \sum_{j < m < i < l}P^{\gamma_n}_{ij}P^{\gamma_n}_{lj}P^{\gamma_n}_{lm}P^{\gamma_n}_{im}\\
	        =& \sum_{j < i}(P^{\gamma_n}_{ij})^4 + \sum_{j < l < i}(P^{\gamma_n}_{ij})^2(P^{\gamma_n}_{lj})^2 + \sum_{j < i < l}(P^{\gamma_n}_{ij})^2(P^{\gamma_n}_{lj})^2 \\
	        &+ \sum_{j < m < i}(P^{\gamma_n}_{ij})^2(P^{\gamma_n}_{im})^2 + \sum_{m < j < i }(P^{\gamma_n}_{ij})^2(P^{\gamma_n}_{im})^2  \\
	        &+ 4\sum_{i < j < l < m}P^{\gamma_n}_{li}P^{\gamma_n}_{mi}P^{\gamma_n}_{lj}P^{\gamma_n}_{mj}\\
	        =& \sum_{j < i}(P^{\gamma_n}_{ij})^4 + 2\sum_{i < j < l}(P^{\gamma_n}_{li})^2(P^{\gamma_n}_{ji})^2 \\
	        &+ \sum_{i < j < l}(P^{\gamma_n}_{li})^2(P^{\gamma_n}_{lj})^2 + 4\sum_{i < j < l < m}P^{\gamma_n}_{li}P^{\gamma_n}_{mi}P^{\gamma_n}_{lj}P^{\gamma_n}_{mj}\\
	        =& \sum_{i < j}(P^{\gamma_n}_{ij})^4 + 2 \sum_{i < j < l}\left((P^{\gamma_n}_{ij})^2(P^{\gamma_n}_{il})^2 + (P^{\gamma_n}_{il})^2(P^{\gamma_n}_{jl})^2\right) \\ 
	        &+ 4 \sum_{i < j < l < m}P^{\gamma_n}_{il}P^{\gamma_n}_{jl}P^{\gamma_n}_{im}P^{\gamma_n}_{jm}.
	    \end{split}
	\end{align}
	The third equality follows by replacing the indices in the summation according to the dictionary $\{i:j, j:i, l:m, m:l\}$. The fourth equality follows since the product within the summation is non-zero only when both $i > j, m$ and $l > j, m$. The fifth equality follows from replacing the indices in the last four summations in the previous line according to the dictionaries $\{i: m, j:j, l:l, m:i\}$, $\{i: l, j:j, l:m, m:i\}$, $\{i: m, j:i, l:l, m:j\}$, $\{i: l, j:i, l:m, m:j\}$, respectively. The sixth equality follows from replacing the indices in the second, third, fourth and fifth summations in the previous line according to the dictionaries $\{i:l, j:i, l:j\}$, $\{i:j, j:i, l:l\}$, $\{i:l, j:i, m:j\}$, $\{i:l, j:j, m:i\}$, respectively. The last equality follows by the symmetry of $P^{\gamma_n}$.
	
	Let $\mathfrak{S} \coloneqq \sum_{i < j < l < m}P^{\gamma_n}_{il}P^{\gamma_n}_{jl}P^{\gamma_n}_{im}P^{\gamma_n}_{jm} + P^{\gamma_n}_{ij}P^{\gamma_n}_{jl}P^{\gamma_n}_{im}P^{\gamma_n}_{lm} + P^{\gamma_n}_{ij}P^{\gamma_n}_{il}P^{\gamma_n}_{jm}P^{\gamma_n}_{lm}$. Substituting the expressions in Equations \eqref{Equation-Trace-L4}--\eqref{Equation-Trace-LTLLTL} into Equation \eqref{Equation-Expansion-L} yields
	\begin{equation*}
	    \begin{aligned}
	    \text{tr}((P^{\gamma_n} - G)^4) =& 8\sum_{i < j < l < m}P^{\gamma_n}_{ij}P^{\gamma_n}_{jl}P^{\gamma_n}_{im}P^{\gamma_n}_{lm}\\
	    & + 4\left(\sum_{i < j < l}(P^{\gamma_n}_{ij})^2(P^{\gamma_n}_{jl})^2 + 2\sum_{i < j < l < m}P^{\gamma_n}_{ij}P^{\gamma_n}_{il}P^{\gamma_n}_{jm}P^{\gamma_n}_{lm}\right)\\
	    & + 2\left(\sum_{i < j}(P^{\gamma_n}_{ij})^4 + 2 \sum_{i < j < l}\left((P^{\gamma_n}_{ij})^2(P^{\gamma_n}_{il})^2 + (P^{\gamma_n}_{il})^2(P^{\gamma_n}_{jl})^2\right)\right.\\
	    & \text{\phantom{ ddddd }}+\left. 4 \sum_{i < j < l < m}P^{\gamma_n}_{il}P^{\gamma_n}_{jl}P^{\gamma_n}_{im}P^{\gamma_n}_{jm}\right)\\
	    =&2\sum_{i < j}(P^{\gamma_n}_{ij})^4\\
	    &+ 4\sum_{i < j < l}\left((P^{\gamma_n}_{ij})^2(P^{\gamma_n}_{jl})^2 + (P^{\gamma_n}_{il})^2(P^{\gamma_n}_{jl})^2 + (P^{\gamma_n}_{ij})^2(P^{\gamma_n}_{il})^2\right)\\
	    &+ 8\mathfrak{S}.
	    \end{aligned}
	\end{equation*}

Next, by the triangle inequality and Lemma \ref{Lemma-P-Linear-Algebra} \ref{Lemma-B1},
	\begin{equation}
	    \label{Equation-Frak-S}
	    \begin{aligned}
		|\mathfrak S| \leq & \frac 1 4 \sum_{i < j}(P^{\gamma_n}_{ij})^4\\
		                   &+ \frac 1 2\sum_{i < j < l}\left((P^{\gamma_n}_{ij})^2(P^{\gamma_n}_{jl})^2 + (P^{\gamma_n}_{il})^2(P^{\gamma_n}_{jl})^2 + (P^{\gamma_n}_{ij})^2(P^{\gamma_n}_{il})^2\right)\\
		                   &+ \frac 1 8 \text{tr}((P^{\gamma_n} - G)^4)\\
		                   \leq & Cr_n.
		\end{aligned}
	\end{equation}
	
	Take $\{u_i\}$ to be a sequence of i.i.d. mean-zero and unit variance random variables. Define
	\begin{equation*}
		\begin{aligned}
		&\Delta_1 \coloneqq \sum_{i < j < l} \left(P^{\gamma_n}_{ij}P^{\gamma_n}_{il}u_{j}u_{l} + P^{\gamma_n}_{ij}P^{\gamma_n}_{jl}u_{i}u_{l} + P^{\gamma_n}_{il}P^{\gamma_n}_{jl}u_{i}u_{j} \right),\\
		&\Delta_2 \coloneqq \sum_{i < j < l} \left(P^{\gamma_n}_{ij}P^{\gamma_n}_{il}u_{j}u_{l} + P^{\gamma_n}_{ij}P^{\gamma_n}_{jl}u_{i}u_{l} \right), \Delta_3 \coloneqq \sum_{i < j < l} \left(P^{\gamma_n}_{il}P^{\gamma_n}_{jl}u_{i}u_{j}\right).
		\end{aligned}
	\end{equation*}
	Then by Lemma \ref{Lemma-P-Linear-Algebra} \ref{Lemma-B1},
	\begin{equation*}
	    \begin{aligned}
		\mathbb{E}[\Delta_3^2] &= \sum_{i < j < l}(P^{\gamma_n}_{il})^2(P^{\gamma_n}_{jl})^2 + 2\sum_{i < j < l < m}P^{\gamma_n}_{il}P^{\gamma_n}_{jl}P^{\gamma_n}_{im}P^{\gamma_n}_{jm}\\
		& \leq r_n + 2\sum_{i < j < l < m}P^{\gamma_n}_{il}P^{\gamma_n}_{jl}P^{\gamma_n}_{im}P^{\gamma_n}_{jm}.
		\end{aligned}
	\end{equation*}
	Furthermore,
	\begin{equation*}
		\mathbb{E}[\Delta_2\Delta_3] = \sum_{i < j < l < m}P^{\gamma_n}_{ij}P^{\gamma_n}_{il}P^{\gamma_n}_{jm}P^{\gamma_n}_{lm} +  \sum_{i < j < l < m}P^{\gamma_n}_{ij}P^{\gamma_n}_{jl}P^{\gamma_n}_{im}P^{\gamma_n}_{lm},
	\end{equation*}
	and
	\begin{equation*}
	    \begin{aligned}
	    \mathbb{E}[\Delta_2^2] =& \sum_{\{i, l\} < j < k}P^{\gamma_n}_{ij}P^{\gamma_n}_{il}P^{\gamma_n}_{mj}P^{\gamma_n}_{ml} + \sum_{i < \{j, l\} < m}P^{\gamma_n}_{ij}P^{\gamma_n}_{jm}P^{\gamma_n}_{il}P^{\gamma_n}_{lm}\\
	                            &+ \sum_{i < j < l < m}P^{\gamma_n}_{ij}P^{\gamma_n}_{im}P^{\gamma_n}_{jl}P^{\gamma_n}_{lm} + \sum_{l < i < j < m}P^{\gamma_n}_{ij}P^{\gamma_n}_{jm}P^{\gamma_n}_{li}P^{\gamma_n}_{lm}\\
	                           =& \sum_{i < j < l}(P^{\gamma_n}_{ij})^2(P^{\gamma_n}_{il})^2 + \sum_{i < j < l}(P^{\gamma_n}_{ij})^2(P^{\gamma_n}_{jl})^2 + 2\sum_{i < m < j < l}P^{\gamma_n}_{ij}P^{\gamma_n}_{il}P^{\gamma_n}_{mj}P^{\gamma_n}_{ml}\\
	                           &+ 2\sum_{i < j < l < m}P^{\gamma_n}_{ij}P^{\gamma_n}_{jm}P^{\gamma_n}_{il}P^{\gamma_n}_{lm}\\
	                           & + \sum_{i < j < l < m}P^{\gamma_n}_{ij}P^{\gamma_n}_{im}P^{\gamma_n}_{jl}P^{\gamma_n}_{lm} + \sum_{i < j < l  < m}P^{\gamma_n}_{jl}P^{\gamma_n}_{lm}P^{\gamma_n}_{ij}P^{\gamma_n}_{im}\\
	                           =& \sum_{i < j < l}(P^{\gamma_n}_{ij})^2(P^{\gamma_n}_{il})^2 + \sum_{i < j < l}(P^{\gamma_n}_{ij})^2(P^{\gamma_n}_{jl})^2 + 2 \mathfrak S\\
	                           & \leq Cr_n,
	    \end{aligned}
	\end{equation*}
	where the last inequality follows from Lemma \ref{Lemma-P-Linear-Algebra} \ref{Lemma-B1} and Equation \eqref{Equation-Frak-S}.
	Since $\Delta_1 = \Delta_2 + \Delta_3$,  $\mathbb{E}[\Delta_1^2] = \mathbb{E}[\Delta_2^2] + \mathbb{E}[\Delta_3^2] + 2\mathbb{E}[\Delta_2\Delta_3] \leq Cr_n + 2\mathfrak S \leq Cr_n$. Hence by the triangle inequality and the expression for $\mathbb{E}[\Delta_3^2]$,
	\begin{equation*}
		\begin{aligned}
		\left|\sum_{i < j < l < m}P^{\gamma_n}_{il}P^{\gamma_n}_{jl}P^{\gamma_n}_{im}P^{\gamma_n}_{jm}\right| &\leq C\left( \mathbb{E}[\Delta_3^2] + r_n\right) \leq C\left( \mathbb{E}[(\Delta_1 - \Delta_2)^2] + r_n\right)\\
		& \leq C(\mathbb{E}[\Delta_1^2] + \mathbb{E}[\Delta_2^2] + r_n)\\
		& \leq Cr_n.
		\end{aligned}
	\end{equation*}
\end{proof}

\begin{lemma}
    \label{Lemma-Gamma-Star-Existence}
    Let~$\gamma_{-}>0$ and~$\Gamma_n=\Gamma(\gamma_{-}) \coloneqq \{\gamma_n\in \Re : \gamma_n \geq 0 \text{ if } r_n = k_n , \text{ and } \gamma_n \geq \gamma_{-} > 0 \text{ if } r_n < k_n\}$. Then if $P^{\gamma_n}$ is a non-diagonal matrix,\footnote{We note that Assumption \ref{Assumption-High-Level} excludes the case where $P^{\gamma_n}$ is diagonal, so that $\gamma^*_n$ exists under the assumptions made in this paper.}
    \begin{equation*}
		\gamma^*_n \coloneqq \max\underset{\gamma_n  \in \Gamma_n}{\arg\max} \sum_{i = 1}^n\sum_{j\neq i}(P^{\gamma_n}_{ij})^2\in[0,\infty)
	\end{equation*}
    exists.
    
    \begin{proof}
    
    Notice that 
    \begin{equation*}
	\begin{aligned}
		\sum_{i = 1}^n\sum_{j\neq i}(P^{\gamma_n}_{ij})^2 &= \sum_{i=1}^n\sum_{j = 1}^n(P^{\gamma_n}_{ij})^2  - \sum_{i = 1}^n(P^{\gamma_n}_{ii})^2\\
		&= \sum_{i = 1}^n(P^{\gamma_n})^2_{ii} -\sum_{i = 1}^n(P^{\gamma_n}_{ii})^2\\
		&= \sum_{l = 1}^{r_n}\left(\frac{D_{ll}^2}{D_{ll}^2 + \gamma_n}\right)^2 - \sum_{i = 1}^n\left(\sum_{l = 1}^{r_n}\frac{D_{ll}^2}{D_{ll}^2 + \gamma_n}U^2_{il}\right)^2,
	\end{aligned}
    \end{equation*}
    where the third equality follows from Lemma \ref{Lemma-P-Linear-Algebra} \ref{Lemma-Traces-2} and Equation \eqref{Equation-P-hh}.
    Since
    \begin{equation*}
    		\underset{\gamma_n \to \infty}{\text{ lim }}\frac{D^2_{ll}}{D^2_{ll} + \gamma_n} = 0,
    \end{equation*}
    and $U^2_{il} \leq 1$ for~$l=1,\hdots,r_n$ and~$i=1,\hdots,n$, it follows that
    \begin{equation*}
    	\begin{aligned}
    	\underset{\gamma_n \to \infty}{\text{ lim }}\sum_{i = 1}^n\sum_{j\neq i}(P^{\gamma_n}_{ij})^2 &= 0.\\
    	\end{aligned}
    \end{equation*}
    Hence, the maximum is not attained for arbitrarily large $\gamma_n$. Since $P^{\gamma_n}$ is a non-diagonal matrix by assumption, $\sum_{i = 1}^n\sum_{j\neq i}(P^{\gamma_n}_{ij})^2$ is strictly positive.  This leaves a compact set over which the non-zero $\sum_{i = 1}^n\sum_{j\neq i}(P^{\gamma_n}_{ij})^2$ is maximised, such that a maximiser exists.


\end{proof}
\end{lemma}

\begin{lemma}
	\label{Lemma-Variance-Cons}
	Under Assumptions \ref{Assumption-MC}, \ref{Assumption-Bekker} and the null hypothesis in Equation \eqref{Equation-RJAR-Null-Hypothesis}, $|\hat{\Phi}_{\gamma_n} - \Phi_{\gamma_n}|\overset{p}{\to} 0$.
	\end{lemma}
	
	\begin{proof}
			Defining $\eta_i:= \varepsilon_i^2 - \mathbb{E}[\varepsilon^2_i]$, one can write 
		\begin{equation*}
			\hat{\Phi}_{\gamma_n} - \Phi_{\gamma_n} = \frac{2}{r_n}\sum_{i = 1}^n\sum_{j\neq i}(P_{ij}^{\gamma_n})^2\left(\eta_i\eta_j + \mathbb{E}[\varepsilon^2_j]\eta_i + \mathbb{E}[\varepsilon^2_i]\eta_j \right),
		\end{equation*}
		 and it follows that
		\begin{equation}
			\label{Equation-RJAR-Abs-Phi-N}
			\begin{aligned}
			\left|\hat{\Phi}_{\gamma_n} - \Phi_{\gamma_n}\right| \leq& \frac{2}{r_n}\left|\sum_{i = 1}^n\sum_{j \neq i}(P_{ij}^{\gamma_n})^2\eta_i\eta_j\right| + \frac{2}{r_n}\left|\sum_{i = 1}^n\sum_{j \neq i}(P_{ij}^{\gamma_n})^2\mathbb{E}[\varepsilon^2_j]\eta_i\right|\\
			&+ \frac{2}{r_n}\left|\sum_{i = 1}^n\sum_{j \neq i}^n(P_{ij}^{\gamma_n})^2\mathbb{E}[\varepsilon^2_i]\eta_j\right|\\
			\equiv& A_1 + A_2 + A_3.
			\end{aligned}
		\end{equation}
		Consider each of $\mathbb{E}[A_1^2]$, $\mathbb{E}[A_2^2]$, and $\mathbb{E}[A_3^2]$ in turn.
		\begin{align*}
			\mathbb{E}[A_1^2] =& \frac{4}{r_n^2}\sum_{i = 1}^n\sum_{j \neq i}\sum_{h = 1}^n\sum_{g \neq h}(P_{ij}^{\gamma_n})^2(P_{hg}^{\gamma_n})^2 \mathbb{E}[\eta_i\eta_j \eta_h\eta_g]\\
			= &   \frac{4}{r_n^2}\sum_{i = 1}^n\sum_{j \neq i}\sum_{h \neq i}\sum_{g \neq h}(P_{ij}^{\gamma_n})^2(P_{hg}^{\gamma_n})^2\mathbb{E}[\eta_i\eta_j \eta_h\eta_g]\\
			&+  \frac{4}{r_n^2}\sum_{i = 1}^n\sum_{j \neq i}\sum_{g \neq i}(P_{ij}^{\gamma_n})^2(P_{ig}^{\gamma_n})^2\mathbb{E}[\eta_i^2\eta_j \eta_g]\\
			= &  \frac{4}{r_n^2}\sum_{i = 1}^n\sum_{j \neq i}\sum_{h \neq i}\sum_{g \notin\{h,i\}}(P_{ij}^{\gamma_n})^2(P_{hg}^{\gamma_n})^2\mathbb{E}[\eta_i]\mathbb{E}[\eta_j \eta_h\eta_g]\\
			+&  \frac{4}{r_n^2}\sum_{i = 1}^n\sum_{j \neq i}\sum_{h \neq i}(P_{ij}^{\gamma_n})^2(P_{ih}^{\gamma_n})^2\mathbb{E}[\eta_i^2\eta_j \eta_h]\\
			&+ \frac{4}{r_n^2}\sum_{i = 1}^n\sum_{j \neq i}\sum_{g \neq i}(P_{ij}^{\gamma_n})^2(P_{ig}^{\gamma_n})^2\mathbb{E}[\eta_i^2\eta_j \eta_g]\\
			= & \frac{8}{r_n^2}\sum_{i = 1}^n\sum_{j \neq i}\sum_{h \neq i}(P_{ij}^{\gamma_n})^2(P_{ih}^{\gamma_n})^2\mathbb{E}[\eta_i^2\eta_j \eta_h]\\
			= & \frac{8}{r_n^2}\sum_{i = 1}^n\sum_{j \neq i}\sum_{h \notin\{i,j\}}(P_{ij}^{\gamma_n})^2(P_{ih}^{\gamma_n})^2\mathbb{E}[\eta_i^2]\mathbb{E}[\eta_j]\mathbb{E}[\eta_h] \\
			&+ \frac{8}{r_n^2}\sum_{i = 1}^n\sum_{j \neq i}(P_{ij}^{\gamma_n})^4\mathbb{E}[\eta_i^2]\mathbb{E}[\eta_j^2]\\
			= & \frac{8}{r_n^2}\sum_{i = 1}^n\sum_{j \neq i}(P_{ij}^{\gamma_n})^4\mathbb{E}[\eta_i^2]\mathbb{E}[\eta_j^2]\\
			\leq & \frac{C }{r_n^2}\sum_{i = 1}^n\sum_{j \neq i}(P_{ij}^{\gamma_n})^4\\
			\leq & \frac{C}{r_n},
		\end{align*}
		The fourth equality follows from $\mathbb{E}[\eta_i] = 0$ and the symmetry of $P^{\gamma_n}$. The sixth equality follows from $\mathbb{E}[\eta_j] = 0$. The first inequality follows from Assumption \ref{Assumption-MC}, which implies $\text{sup}_{i\in \mathbb{N}}\mathbb{E}[\eta_i^2] < \infty$. The second inequality follows from Lemma \ref{Lemma-P-Linear-Algebra} \ref{Lemma-B1}. Next, concerning $\mathbb{E}[A_2^2]$ one has
		\begin{equation*}
			\begin{aligned}
				\mathbb{E}[A_2^2] =& \frac{4}{r_n^2}\sum_{i = 1}^n\sum_{j \neq i} \sum_{h = 1}^n\sum_{g\neq h}(P_{ij}^{\gamma_n})^2(P_{hg}^{\gamma_n})^2\mathbb{E}[\varepsilon^2_j]\mathbb{E}[\varepsilon^2_g]\mathbb{E}[\eta_i\eta_h]\\
				= & \frac{4}{r_n^2}\sum_{i = 1}^n\sum_{j \neq i} \sum_{h \neq i}\sum_{g\neq h}(P_{ij}^{\gamma_n})^2(P_{hg}^{\gamma_n})^2\mathbb{E}[\varepsilon^2_j]\mathbb{E}[\varepsilon^2_g]\mathbb{E}[\eta_i]\mathbb{E}[\eta_h]\\ 
				+ & \frac{4}{r_n^2}\sum_{i = 1}^n\sum_{j \neq i} \sum_{g\neq i}(P_{ij}^{\gamma_n})^2(P_{ig}^{\gamma_n})^2\mathbb{E}[\varepsilon^2_j]\mathbb{E}[\varepsilon^2_g]\mathbb{E}[\eta_i^2]\\
				= & \frac{4}{r_n^2}\sum_{i = 1}^n\sum_{j \neq i} \sum_{g\neq i}(P_{ij}^{\gamma_n})^2(P_{ig}^{\gamma_n})^2\mathbb{E}[\varepsilon^2_j]\mathbb{E}[\varepsilon^2_g]\mathbb{E}[\eta_i^2]\\
				\leq & \frac{C}{r_n^2}\sum_{i = 1}^n\sum_{j \neq i} \sum_{g\neq i}(P_{ij}^{\gamma_n})^2(P_{ig}^{\gamma_n})^2\\
				\leq & \frac{C}{r_n},
			\end{aligned}
		\end{equation*}
		by the same reasoning that led to the penultimate display. Finally, similar arguments imply that $\mathbb{E}[A_3^2] \leq \frac{C}{r_n}$ such that by Markov's inequality and Equation \eqref{Equation-RJAR-Abs-Phi-N} it follows that $\hat{\Phi}_{\gamma_n} - \Phi_{\gamma_n} = O_p({r^{-1/2}_n})$. Thus, by Assumption \ref{Assumption-Bekker}, $|\hat{\Phi}_{\gamma_n} - \Phi_{\gamma_n}|\overset{p}{\to} 0$. 
	\end{proof}

\subsection{Proof of Proposition \ref{Proposition-Weaker-CMSMS}\label{Appendix-Weaker-CMSMS}}
As observed after the statement of Proposition \ref{Proposition-Weaker-CMSMS}, it suffices to prove its part 2. Thus, 
assume that there exists a $\delta' \in (0, 1)$ such that $\frac{1}{k_n}|\mathcal{A}_n(\delta')| \to 0$. Notice that the sum in Assumption \ref{Assumption-High-Level} with $\gamma_n = 0$ can be written as
\begin{equation*}
    \begin{aligned}
       \frac{1}{k_n} \sum_{i = 1}^n\sum_{j\neq i} P_{ij}^2 &= \frac{1}{k_n} \left(\sum_{i = 1}^nP_{ii} - \sum_{i = 1}^nP_{ii}^2\right)\\
                                             &= \frac{1}{k_n} \left(k_n - \sum_{i \in \mathcal{A}^c_n(\delta')} P_{ii}^2 - \sum_{i \in \mathcal{A}_n(\delta')} P_{ii}^2\right)\\
                                             &\geq \frac{1}{k_n} \left(k_n - (1-\delta')\sum_{i \in \mathcal{A}_n^c(\delta')} P_{ii} - \sum_{i \in \mathcal{A}_n(\delta')} P_{ii}^2\right)\\
                                             & \geq \frac{1}{k_n} \left(k_n - (1-{\delta}')k_n - |\mathcal{A}_n(\delta')|\right)\\
                                             & \geq \frac{1}{k_n} \left( k_n\delta' - |\mathcal{A}_n(\delta')|\right),
    \end{aligned}
\end{equation*}
where the first equality uses the idempotency and symmetry of $P$, and the second equality uses the fact that the trace of~$P$ is equal to its rank~$k_n$. Thus, if $\frac{1}{k_n}|\mathcal{A}_n(\delta')| \to 0$, 
\begin{align*}
   \liminf_{n\to\infty}\frac{1}{k_n}\sum_{i = 1}^n\sum_{j\neq i}(P^{\gamma_n^*}_{ij})^2 
   \geq
    \liminf_{n\to\infty} \frac{1}{k_n}\sum_{i = 1}^n\sum_{j\neq i}P_{ij}^2 
    \geq 
    \delta'
    >
   0, 
\end{align*}
and hence Assumption \ref{Assumption-High-Level} is satisfied.

\qed

\subsection{Proof of Proposition \ref{prop:hd}}
The following preparatory lemma is likely well-known but we state it here for ease of reference. Its proof is a simple consequence of the Marchenko-Pastur theorem.\footnote{We thank Alexei Onatski for suggesting this.} Denote by~$\lambda_{n,1}\geq \hdots\geq \lambda_{n,n}\geq 0$ the eigenvalues of~$\frac{ZZ'}{k_n}$, where $\lambda_{n,i}=\frac{D^2_{ii}}{k_n}$ in our earlier notation on p.~\pageref{eq:Pgamma}, and by~$F_n$ the empirical distribution of these. 
\begin{lemma}\label{lem:moments}
	Let~$Z$ be an~$n\times k_n$ matrix with i.i.d.~entries of mean zero and variance one. In addition, let~$p>2$ satisfy~$\mathbb{E}|Z_{11}|^{4p}<\infty$ and let~$k_n=\tau n$ for~$\tau\in[1,\infty)$. Then
	\begin{enumerate}
\item $\frac{1}{n}\sum_{i=1}^n\lambda_{n,i}\to 1$ almost surely.
		\item $\frac{1}{n}\text{tr}\sbr[2]{\del[1]{\frac{ZZ'}{n}}^2}
		=
		 \frac{\tau^2}{n}\text{tr}\sbr[2]{\del[1]{\frac{ZZ'}{k_n}}^2}
		 =
		  \frac{\tau^2}{n}\sum_{i=1}^n\lambda_{n,i}^2
		 \to \tau^2+\tau$ almost surely.
		 \item $\frac{1}{n}\sum_{i=1}^n \sbr[2]{\del[0]{\frac{ZZ'}{n}}_{ii}}^2
		  =
		  \frac{\tau^2}{n}\sum_{i=1}^n \sbr[2]{\del[0]{\frac{ZZ'}{k_n}}_{ii}}^2
		  \to \tau^2
		  $ almost surely.
	\end{enumerate}
\end{lemma}

\begin{proof}
	Denote by~$F_\tau$ the Marchenko-Pastur distribution with ``dimension-to-sample-ratio''~$1/\tau$ (observe that the dimension of~$ZZ'$ is~$n\times n$ whereas each entry is based on a ``sample size'' of~$k_n$). By the Marchenko-Pastur theorem (as stated in Theorem 3.6 of \citet{bai2010spectral}),~$F_n$ almost surely converges weakly to~$F_\tau$. In addition, by the Bai-Yin-Krishnaiah law (as stated in Theorem 5.8 of~\citet{bai2010spectral}) it follows that~$\lambda_{n,1}\to (1+\tau^{-1/2})^2$ almost surely. Hence, outside a set of probability zero for~$n$ sufficiently large~$F_n$ has compact support and therefore~$F_n \circ (x\mapsto x^a)^{-1}$ is uniformly integrable for all~$a\in(0,\infty)$.
	
	We first establish part~1. This follows by the observations just made because
	\begin{align*}
	\frac{1}{n}\sum_{i=1}^n\lambda_{n,i}
	=
	\int x dF_n
	\to 
	\int x dF_\tau
	=
	1\qquad \text{almost surely}, 	
	\end{align*}
	where the last equality follows from, e.g., Lemma 3.1 in \citet{bai2010spectral}.

	 To prove part 2., note that since~$k_n/n=\tau$
	\begin{align*}
	\frac{1}{n}\text{tr}\sbr[3]{\del[2]{\frac{ZZ'}{n}}^2}
		=
		 \frac{\tau^2}{n}\text{tr}\sbr[3]{\del[2]{\frac{ZZ'}{k_n}}^2}
		 =
		 \frac{\tau^2}{n}\sum_{i=1}^n\lambda_{n,i}^2.
		 	\end{align*}
		 	The almost sure convergence to~$\tau^2+\tau$ follows from the observations prior to the previous display, i.e.
		 	\begin{align*}
		 	\frac{\tau^2}{n}\sum_{i=1}^n\lambda_{n,i}^2
		 	=
		 	\tau^2\int x^2 dF_n
		 \to 
		 \tau^2 \int x^2dF_{\tau}
		 =
		 \tau^2(1+1/\tau)=\tau^2+\tau \qquad\text{almost surely},	
		 	\end{align*}
		 	 where the second-to-last equality follows from, e.g., Lemma 3.1 in \citet{bai2010spectral}.
		 	
		 	Next, to establish part 3., note that
		 	\begin{align*}
		 		\mathbb{E}  \sbr[2]{\del[2]{\frac{ZZ'}{k_n}}_{ii}}^2
		 		=
		 		\frac{1}{k_n^2}\sum_{j=1}^{k_n}\mathbb{E}Z_{ij}^4 + \frac{1}{k_n^2}\sum_{j=1}^{k_n}\sum_{l\neq j}^{k_n}\mathbb{E}Z_{ij}^2\mathbb{E}Z_{il}^2
		 		\to 1\qquad\text{for }i=1,\hdots,n 
		 	\end{align*}
		 	and that the~$\sbr[1]{\del[1]{\frac{ZZ'}{k_n}}_{ii}}^2$ are independent across~$i=1,\hdots,n$. Therefore, an application of Markov's inequality along with the Marcinkiewicz-Zygmund inequality yields that there exists a~$c_p>0$ such that for every~$\varepsilon>0$ 
		 	\begin{align*}
		 		\mathbb{P}\del[3]{ \envert[3]{\frac{1}{n}\sum_{i=1}^n \sbr[2]{\del[2]{\frac{ZZ'}{k_n}}_{ii}}^2-\mathbb{E}  \sbr[2]{\del[2]{\frac{ZZ'}{k_n}}_{11}}^2}\geq \varepsilon}
		 	&\leq 	
		 	c_p\frac{\mathbb{E}\del[3]{\sum_{i=1}^n \del[2]{\sbr[2]{\del[2]{\frac{ZZ'}{k_n}}_{ii}}^2-\mathbb{E}  \sbr[2]{\del[2]{\frac{ZZ'}{k_n}}_{11}}^2}^2}^{p/2}}{\varepsilon^p n^p}\\
		 	&= 
		 	c_p\frac{n^{p/2}\mathbb{E}\del[3]{\frac{1}{n}\sum_{i=1}^n \del[2]{\sbr[2]{\del[2]{\frac{ZZ'}{k_n}}_{ii}}^2-\mathbb{E}  \sbr[2]{\del[2]{\frac{ZZ'}{k_n}}_{11}}^2}^2}^{p/2}}{\varepsilon^p n^p} \\
		 	&\leq 
		 	c_p\frac{n^{p/2}\mathbb{E}\envert[2]{\sbr[2]{\del[2]{\frac{ZZ'}{k_n}}_{11}}^2-\mathbb{E}  \sbr[2]{\del[2]{\frac{ZZ'}{k_n}}_{11}}^2}^{p}}{\varepsilon^p n^p},
		 	\end{align*} 
		 	where the second inequality follows by Jensen's inequality as well as the~$\sbr[1]{\del[1]{\frac{ZZ'}{k_n}}_{ii}}^2$ being identically distributed. Another application of Jensen's inequality shows that the right-hand side is bounded from above by~$c(p,\varepsilon)n^{-p/2}$ [here $c(p,\varepsilon)$ is a non-negative constant depending only on~$p$ and~$\varepsilon$]. Therefore, part 3.~follows from the Borel-Cantelli lemma.
		 	\end{proof}

    \begin{proof}[Proof of Proposition~\ref{prop:hd}]
	Since the entries of~$Z$ have a distribution that is absoluteley continuous with respect to the Lebesgue measure, the rank of~$Z$ is~$n$ with probability one. Hence, we show that there exists a~$\gamma_{-}>0$ and a sequence~$\gamma_n$ satisfying~$\gamma_n\geq \gamma_{-}$ for all~$n\in\mathbb{N}$ such that 
	\begin{align}\label{eq:Ass3show}
		\liminf_{n\to\infty}\frac{1}{n}\sum_{i=1}^n\sum_{j\neq i}^n\del[1]{P_{ij}^{\gamma_n}}^2
		>
		0 \qquad \text{almost surely}.
	\end{align}
	Observe that
	\begin{align*}
	\frac{1}{n}\sum_{i=1}^n\sum_{j\neq i}^n\del[1]{P_{ij}^{\gamma_n}}^2
=
\frac{1}{n}\sum_{i=1}^n\del[1]{P^{\gamma_n}}^2_{ii}
-
\frac{1}{n}\sum_{i=1}^n\del[1]{P^{\gamma_n}_{ii}}^2
=
	\frac{1}{n}\text{tr}\sbr[1]{\del[1]{P^{\gamma_n}}^2}
	-
	\frac{1}{n}\sum_{i=1}^n\del[1]{P^{\gamma_n}_{ii}}^2.
	\end{align*}
	Therefore,~\eqref{eq:Ass3show} will follow if we establish the existence of a sequence~$\gamma_n$, bounded away from zero, such that almost surely
	\begin{align}\label{eq:auxshow}
		\liminf_{n\to\infty}\frac{1}{n}\text{tr}\sbr[1]{\del[1]{P^{\gamma_n}}^2}
		>
		\limsup_{n\to\infty}\frac{1}{n}\sum_{i=1}^n\del[1]{P^{\gamma_n}_{ii}}^2.
	\end{align} 
	It follows by \eqref{eq:Pgamma} and its ensuing discussion that
	\begin{align}\label{eq:eigendecomp}
		\frac{1}{\tau}P^{\gamma_n}
		=
		U\frac{1}{\tau}\tilde{D}U'
		=
		U\check{D}U',
	\end{align}
	with~$\check{D}=\frac{1}{\tau}\tilde{D}$ such that with $\gamma_n=\eta_nn$ for some~$\eta_n\geq 0$
	\begin{align}\label{eq:eigen}
	\check{D}_{ii}
	=
	\frac{1}{\tau}\frac{\lambda_{n,i}}{\lambda_{n,i}+\gamma_n/k_n}
	=
	\frac{\lambda_{n,i}}{\tau\lambda_{n,i}+\eta_n}
	\qquad\text{for } i=1,\hdots n.
	\end{align}
	We begin by lower bounding the left-hand side of~\eqref{eq:auxshow}.
	\begin{align*}
	\frac{1}{n}\text{tr}\sbr[1]{\del[1]{P^{\gamma_n}}^2}
	&\geq
	\frac{1}{n}\text{tr}\sbr[2]{\del[2]{\frac{1}{\eta_n n}ZZ'}^2}
	-\envert[3]{\frac{1}{n}\text{tr}\sbr[2]{\del[2]{\frac{1}{\eta_n n}ZZ'}^2}-\frac{1}{n}\text{tr}\sbr[1]{\del[1]{P^{\gamma_n}}^2}}\\
	&
	=
	\frac{\tau^2}{n}\text{tr}\sbr[2]{\del[2]{\frac{1}{\eta_n k_n}ZZ'}^2}
	-
	\frac{\tau^2}{n}\envert[3]{\text{tr}\sbr[2]{\del[2]{\frac{1}{\eta_n k_n}ZZ'}^2}-\text{tr}\sbr[2]{\del[2]{\frac{1}{\tau}P^{\gamma_n}}^2}}\\
	&=
	\frac{\tau^2}{n\eta_n^2}\sum_{i=1}^n\lambda_{n,i}^2
	-
	\frac{\tau^2}{n}\envert[3]{\frac{1}{\eta_n^2}\sum_{i=1}^n\lambda_{n,i}^2-\sum_{i=1}^n\frac{\lambda_{n,i}^2}{(\tau\lambda_{n,i}+\eta_n)^2}}	
	\end{align*}
	Furthermore,
	\begin{align*}
	\envert[3]{\frac{1}{\eta_n^2}\sum_{i=1}^n\lambda_{n,i}^2-\sum_{i=1}^n\frac{\lambda_{n,i}^2}{(\tau\lambda_{n,i}+\eta_n)^2}}
	=
	\frac{1}{\eta_n^2}\sum_{i=1}^n\lambda_{n,i}^2\sbr[3]{1-\frac{\eta_n^2}{(\tau\lambda_{n,i}+\eta_n)^2}},	
	\end{align*}
	where by Theorem 5.8 of~\citet{bai2010spectral} for~$i=1,\hdots,n$
	\begin{align*}
		\lambda_{n,i}\leq \lambda_{n,1}\to (1+\tau^{-1/2})^2\qquad \text{almost surely}
	\end{align*}
	such that for any~$\varepsilon>0$ we can and do choose~$\eta_n= \frac{2\tau(1+\tau^{-1/2})^2}{\varepsilon}=:\eta_\varepsilon$ not depending on~$n$  to ensure that~$\lambda_{n,1}\leq \varepsilon\tau^{-1} \eta_n$ eventually. Thus, almost surely eventually	
	\begin{align*}
		\envert[3]{\frac{1}{\eta_\varepsilon^2}\sum_{i=1}^n\lambda_{n,i}^2-\sum_{i=1}^n\frac{\lambda_{n,i}^2}{(\tau\lambda_{n,i}+\eta_\varepsilon)^2}}
		\leq
		\frac{1}{\eta_\varepsilon^2}\sum_{i=1}^n\lambda_{n,i}^2\sbr[3]{1-\frac{1}{(1+\varepsilon)^2}}.
	\end{align*}
	and it follows by part 2.~of Lemma~\ref{lem:moments} that with~$\gamma_{n,\varepsilon}=n\eta_\varepsilon$
	\begin{align}\label{eq:P1}
	\liminf_{n\to\infty}\frac{1}{n}\text{tr}\sbr[1]{\del[1]{P^{\gamma_{n,\varepsilon}}}^2}	
	\geq
	\liminf_{n\to\infty}\frac{\tau^2}{(1+\varepsilon)^2n\eta_\varepsilon^2}\sum_{i=1}^n\lambda_{n,i}^2 = \frac{\tau^2+\tau}{(1+\varepsilon)^2\eta_\varepsilon^2}\qquad \text{almost surely},
	\end{align}
We now bound the right-hand side of~\eqref{eq:auxshow} from above. To this end, observe that for~$i=1,\hdots,n$
\begin{align*}
	\frac{1}{n}\sum_{i=1}^n\del[1]{P^{\gamma_{n,\varepsilon}}_{ii}}^2
	&\leq
	\frac{1}{n}\sum_{i=1}^n\sbr[2]{\del[2]{\frac{ZZ'}{\eta_\varepsilon n}}_{ii}}^2+\envert[3]{\frac{1}{n}\sum_{i=1}^n\sbr[2]{\del[2]{\frac{ZZ'}{\eta_\varepsilon n}}_{ii}}^2-\frac{1}{n}\sum_{i=1}^n\del[1]{P^{\gamma_{n,\varepsilon}}_{ii}}^2}\\
	&=
		\frac{\tau^2}{n}\sum_{i=1}^n\sbr[2]{\del[2]{\frac{ZZ'}{\eta_\varepsilon k_n}}_{ii}}^2+\frac{\tau^2}{n}\envert[3]{\sum_{i=1}^n\sbr[2]{\del[2]{\frac{ZZ'}{\eta_\varepsilon k_n}}_{ii}}^2-\sum_{i=1}^n\del[2]{\frac{1}{\tau}P^{\gamma_{n,\varepsilon}}_{ii}}^2}
\end{align*}
with
\begin{align*}
	\envert[3]{\sum_{i=1}^n\sbr[2]{\del[2]{\frac{ZZ'}{\eta_\varepsilon k_n}}_{ii}}^2-\sum_{i=1}^n\del[2]{\frac{1}{\tau}P^{\gamma_{n,\varepsilon}}_{ii}}^2}
	\leq
	\max_{1\leq i\leq n}\sbr[2]{\del[2]{\frac{ZZ'}{\eta_\varepsilon k_n}}_{ii}+\frac{1}{\tau}P^{\gamma_{n,\varepsilon}}_{ii}}\sum_{i=1}^n\envert[3]{\del[2]{\frac{ZZ'}{\eta_\varepsilon k_n}}_{ii}-\frac{1}{\tau}P^{\gamma_{n,\varepsilon}}_{ii}}.
\end{align*}
By~\eqref{eq:eigendecomp} and \eqref{eq:eigen} 
\begin{align*}
\del[2]{\frac{ZZ'}{\eta_\varepsilon k_n}}-\frac{1}{\tau}P^{\gamma_n}
=
U\del[1]{\eta_\varepsilon^{-1}\text{diag}(\lambda_{n,1},\hdots,\lambda_{n,n})-\check{D}}U',
\end{align*}
and since
\begin{align*}
\del[1]{\eta_\varepsilon^{-1}\text{diag}(\lambda_{n,1},\hdots,\lambda_{n,n})-\check{D}}_{ii}
=
\frac{\lambda_{n,i}}{\eta_\varepsilon}-\frac{\lambda_{n,i}}{\tau\lambda_{n,i}+\eta_\varepsilon}
=
\frac{\lambda_{n,i}}{\eta_\varepsilon}\del[2]{1-\frac{\eta_\varepsilon}{\tau\lambda_{n,i}+\eta_\varepsilon}}
\geq 
0
	\qquad\text{for } i=1,\hdots n
\end{align*}
it follows that all eigenvalues of~$\del[1]{\frac{ZZ'}{\eta_\varepsilon k_n}}-\frac{1}{\tau}P^{\gamma_{n,\varepsilon}}$ are non-negative such that it is positive semi-definite. Hence, all its diagonal elements are non-negative; that is~$\frac{1}{\tau}P^{\gamma_{n,\varepsilon}}_{ii}\leq \del[1]{\frac{ZZ'}{\eta_\varepsilon k_n}}_{ii}$ for~$i=1,\hdots,n$.
Since the diagonal elements of a symmetric matrix are bounded from above by the largest eigenvalue, it follows  Theorem 5.8 of~\citet{bai2010spectral} that, almost surely,
\begin{align*}
\limsup_{n\to\infty}\max_{1\leq i\leq n}\sbr[2]{\del[2]{\frac{ZZ'}{\eta_\varepsilon k_n}}_{ii}+\frac{1}{\tau}P^{\gamma_{n,\varepsilon}}_{ii}}
\leq
\limsup_{n\to\infty}\max_{1\leq i\leq n}2\del[2]{\frac{ZZ'}{\eta_\varepsilon k_n}}_{ii}
\leq 
\limsup_{n\to\infty}\frac{2\lambda_{n,1}}{\eta_\varepsilon} 
=
\frac{2(1+\tau^{-1/2})^2}{\eta_\varepsilon}.
\end{align*}
Next, recalling that~$\frac{1}{\tau}P^{\gamma_{n,\varepsilon}}_{ii}\leq \del[1]{\frac{ZZ'}{\eta_\varepsilon k_n}}_{ii}$ for~$i=1,\hdots,n$, it follows  by the penultimate display, 
\begin{align*}
	\sum_{i=1}^n\envert[3]{\del[2]{\frac{ZZ'}{\eta_\varepsilon k_n}}_{ii}-\frac{1}{\tau}P^{\gamma_{n,\varepsilon}}_{ii}}
	=
	\sum_{i=1}^n\frac{\lambda_{n,i}}{\eta_\varepsilon}\del[2]{1-\frac{\eta_\varepsilon}{\tau\lambda_{n,i}+\eta_\varepsilon}}
	\leq
	\sum_{i=1}^n\frac{\lambda_{n,i}}{\eta_\varepsilon}\del[2]{1-\frac{1}{1+\varepsilon}},
\end{align*}
almost surely (for~$n$ sufficiently large) implying that
\begin{align*}
	\limsup_{n\to\infty}\frac{1}{n}\sum_{i=1}^n\envert[3]{\del[2]{\frac{ZZ'}{\eta_\varepsilon k_n}}_{ii}-\frac{1}{\tau}P^{\gamma_{n,\varepsilon}}_{ii}}
	\leq
	\frac{1}{\eta_\varepsilon}\del[2]{1-\frac{1}{1+\varepsilon}}\qquad \text{almost surely}
\end{align*}
by part 1.~of Lemma \ref{lem:moments}. Thus, in total, and using part 3.~of Lemma~\ref{lem:moments},
\begin{align}\label{eq:P2}
	\limsup_{n\to\infty}\frac{1}{n}\sum_{i=1}^n\del[1]{P^{\gamma_{n,\varepsilon}}_{ii}}^2
	\leq 
\frac{\tau^2}{\eta_\varepsilon^2}\sbr[2]{1+2(1+\tau^{-1/2})^2\del[2]{1-\frac{1}{1+\varepsilon}}}\qquad \text{almost surely}.	
\end{align}
For any~$\delta>0$ one can ensure that the right-hand side in the previous display does not exceed~$\frac{\tau^2}{\eta_\varepsilon^2}(1+\delta)$ by choosing~$\varepsilon$ sufficiently close to zero. Thus, together~\eqref{eq:P1} and~\eqref{eq:P2} yield~\eqref{eq:auxshow} (with~$\gamma_{n,\varepsilon}$ being the sought after sequence~$\gamma_n$). 
\end{proof}

\subsection{Proof of Theorem \ref{Theorem-RJAR}  \label{Appendix-RJAR}}


We begin by showing that 
\begin{align}\label{eq: RJAR-Known-Variance}
	{RJAR}^*_{\gamma_n} & \coloneqq   \frac{1}{\sqrt{r_n}\sqrt {{\Phi}_{\gamma_n}}}\sum_{i = 1}^n \sum_{j \neq i}P_{ij}^{\gamma_n}\varepsilon_i\varepsilon_j
\end{align}
converges in distribution to a standard normal {for all sequences of $\gamma_n$ that satisfy Assumption \ref{Assumption-High-Level}}.\footnote{We do not directly invoke Lemma 2 in \citet{Hansen:2014ie} because it would require strengthening the assumptions used in this paper.} To this end, define $\mathcal{U}_n := 2\sum_{i = 2}^n\sum_{j = 1}^{i-1}P^{\gamma_n}_{ij}\varepsilon_i\varepsilon_j$, $s_n^2: = \mathbb{E}[\mathcal{U}_n^2]$ and note that by the symmetry of $P^{\gamma_n}$, $RJAR^*_{\gamma_n}$ can be written as
	\begin{equation*}
		RJAR^*_{\gamma_n} = s_n^{-1}\mathcal{U}_n.
	\end{equation*}
We proceed by establishing that $s_n^{-1}\mathcal{U}_n \overset{d}{\to} \mathcal{N}[0, 1]$ as $n\to\infty$. This, in turn, follows from \citet[Corollary 3.1]{Hall:1980uh}) upon verifying that i) for all $\epsilon>0$
	\begin{equation}
		\label{Equation-Hall-Cond-1}
		s_n^{-2}\sum_{i = 2}^n\mathbb{E}[Y^2_{ni}I(|Y_{ni}|>\epsilon s_n)]\to 0\qquad \text{as } n\to\infty,
	\end{equation} 
	where $Y_{ni} = 2\sum_{j = 1}^{i-1}P^{\gamma_n}_{ij}\varepsilon_i\varepsilon_j$, and ii)
	\begin{equation}
		\label{Equation-Hall-Cond-2}
		s_n^{-2}\mathcal{V}_n^2 \overset{p}{\to } 1\qquad\text{as } n\to\infty,	
	\end{equation}
	where $\mathcal{V}_n^2 = \sum_{i = 2}^n\mathbb{E}[Y^2_{ni}|\varepsilon_1, \dots, \varepsilon_{i-1}]$.
Consider first the condition in Equation \eqref{Equation-Hall-Cond-1} and write
\begin{equation*}
	\mathbb{E}[Y^2_{ni}] = 4\sum_{j = 1}^{i-1}\sum_{h = 1}^{i-1}P_{ij}^{\gamma_n} P_{ih}^{\gamma_n} \mathbb{E}[\varepsilon_i^2\varepsilon_j\varepsilon_h] = 4\sum_{j = 1}^{i-1}(P_{ij}^{\gamma_n})^2\mathbb{E}[\varepsilon_i^2]\mathbb{E}[\varepsilon_j^2],
\end{equation*}
so that, upon using that $\mathcal{U}_n=\sum_{i=2}^n Y_{ni}$ and $\mathbb{E}[Y_{ni}Y_{nj}]=0$ for $i\neq j$, one gets
\begin{equation}\label{eq:snaux}
	s_n^2 = \sum_{i = 2}^n\mathbb{E}[Y^2_{ni}] = 4\sum_{i = 2}^n\sum_{j = 1}^{i-1}(P_{ij}^{\gamma_n})^2\mathbb{E}[\varepsilon_i^2]\mathbb{E}[\varepsilon_j^2].
\end{equation}
Furthermore,
\begin{equation*}
	\begin{aligned}
	\mathbb{E}[Y_{ni}^4] =& 16\sum_{j = 1}^{i-1}\sum_{h = 1}^{i-1}\sum_{m = 1}^{i-1}\sum_{l = 1}^{i-1}P_{ij}^{\gamma_n} P_{ih}^{\gamma_n} P_{im}^{\gamma_n} P_{il}^{\gamma_n} \mathbb{E}[\varepsilon_i^4\varepsilon_j\varepsilon_h\varepsilon_m\varepsilon_l]\\
	=& 16\sum_{j = 1}^{i-1}(P_{ij}^{\gamma_n})^4\mathbb{E}[\varepsilon_i^4]\mathbb{E}[\varepsilon_j^4] + 48\sum_{j = 1}^{i-1}\sum_{h\neq j}(P_{ij}^{\gamma_n})^2(P_{ih}^{\gamma_n})^2\mathbb{E}[\varepsilon_i^4]\mathbb{E}[\varepsilon_j^2]\mathbb{E}[\varepsilon_h^2],
	\end{aligned}
\end{equation*}
so that
\begin{equation*}
    \begin{aligned}
	\sum_{i = 2}^n\mathbb{E}[Y_{ni}^4] =&16\sum_{i = 2}^n\sum_{j = 1}^{i-1}(P_{ij}^{\gamma_n})^4\mathbb{E}[\varepsilon_i^4]\mathbb{E}[\varepsilon_j^4]\\
	&+ 48\sum_{i = 2}^n\sum_{j = 1}^{i-1}\sum_{h\neq j}(P_{ij}^{\gamma_n})^2(P_{ih}^{\gamma_n})^2\mathbb{E}[\varepsilon_i^4]\mathbb{E}[\varepsilon_j^2]\mathbb{E}[\varepsilon_h^2].
	\end{aligned}
\end{equation*}
It follows that
\begin{equation}
	\label{Equation-Weird-Condition}	
	\begin{aligned}
		\frac{\sum_{i = 2}^n\mathbb{E}[Y_{ni}^4]}{s^4_n} &\leq C  \frac{\sum_{i = 2}^n\sum_{j = 1}^{i-1}(P_{ij}^{\gamma_n})^4 +\sum_{i = 2}^n\sum_{j = 1}^{i-1}\sum_{h\neq j}(P_{ij}^{\gamma_n})^2(P_{ih}^{\gamma_n})^2}{\left(\sum_{i = 1}^n\sum_{j\neq i}(P^{\gamma_n}_{ij})^2\right)^2}\\
		&\leq  C\frac{\sum_{i = 1}^n\sum_{j = 1}^{n}(P_{ij}^{\gamma_n})^2 +\sum_{i = 1}^n\sum_{j = 1}^{n}(P_{ij}^{\gamma_n})^2\sum_{h = 1}^n(P_{ih}^{\gamma_n})^2}{\left(\sum_{i = 1}^n\sum_{j\neq i}(P^{\gamma_n}_{ij})^2\right)^2}\\
		&=  C\frac{\sum_{i = 1}^n(P^{\gamma_n})^2_{ii} +\sum_{i = 1}^n\left((P^{\gamma_n})^2_{ii}\right)^2}{\left(\sum_{i = 1}^n\sum_{j\neq i}(P^{\gamma_n}_{ij})^2\right)^2}\\
		&\leq  C\frac{\sum_{i = 1}^n(P^{\gamma_n})^2_{ii} +\sum_{i = 1}^n(P^{\gamma_n})^2_{ii}}{\left(\sum_{i = 1}^n\sum_{j\neq i}(P^{\gamma_n}_{ij})^2\right)^2}\\
		&\leq C  \frac{\sum_{i = 1}^n(P^{\gamma_n})^2_{ii}}{r^{2}_n} \\
		&\leq Cr^{1 - 2}_n \\
		&\to 0,
	\end{aligned}
\end{equation}
where the first inequality follows from Assumption \ref{Assumption-MC}. The second inequality follows from Lemma \ref{Lemma-P-Linear-Algebra} \ref{Lemma-One-Bounds}. The third inequality follows from Lemma \ref{Lemma-P-Linear-Algebra} \ref{Lemma-P-Squared-Ineq}. The fourth inequality follows from Equation \eqref{Equation-Assumption-High-Level-Implication}. The fifth inequality follows from Lemma \ref{Lemma-P-Linear-Algebra} \ref{Lemma-Traces-1}. The limit holds by Assumption \ref{Assumption-Bekker} and Assumption \ref{Assumption-High-Level}. The condition in Equation \eqref{Equation-Hall-Cond-1} follows from Equation \eqref{Equation-Weird-Condition}. 
%
%
%
%

In order to verify the convergence in Equation \eqref{Equation-Hall-Cond-2} it suffices to show that
\begin{equation*}
	s_n^{-4}\mathbb{E}\left[(\mathcal{V}_n^2 - s_n^2)^2\right] = \frac{\mathbb{E}[\mathcal{V}_n^4] + s_n^4 - 2s_n^2\mathbb{E}[\mathcal{V}_n^2]}{s_n^4} = \frac{\mathbb{E}[\mathcal{V}_n^4]}{s_n^4} + 1 - 2\frac{\mathbb{E}[\mathcal{V}_n^2]}{s_n^2} \to 0. 
\end{equation*}
Since $\mathcal{V}_n^2 = \sum_{i = 2}^n\mathbb{E}[Y^2_{ni}|\varepsilon_1, \dots, \varepsilon_{i-1}]$ it follows that
\begin{align*}
    \mathbb{E}[\mathcal{V}_n^2]
    =
    \sum_{i=2}^n \mathbb{E}[Y_{ni}^2]
    =
    s_n^2,
\end{align*}
the last equality following from Equation \eqref{eq:snaux}. It remains to be verified that 
\begin{equation}
	\label{Equation-Weird-Condition-2}
	\frac{\mathbb{E}[\mathcal{V}_n^4]}{s_n^4}\to 1.
\end{equation}
To this end, observe that
\begin{equation*}
	\mathcal{V}_n^2 = \sum_{i = 2}^n\mathbb{E}[Y^2_{ni}|\varepsilon_1, \dots, \varepsilon_{i-1}] = 4\sum_{i = 2}^n\sum_{j = 1}^{i-1}\sum_{h = 1}^{i-1}P_{ij}^{\gamma_n} P_{ih}^{\gamma_n} \mathbb{E}[\varepsilon_i^2]\varepsilon_j\varepsilon_h,
\end{equation*}
such that
\begin{equation*}
	\mathcal{V}_n^4 = 16\sum_{i = 2}^n\sum_{j = 2}^n\sum_{h = 1}^{i-1}\sum_{l = 1}^{i-1}\sum_{m = 1}^{j-1}\sum_{w = 1}^{j-1}P^{\gamma_n}_{ih}P^{\gamma_n}_{il}P^{\gamma_n}_{jm}P^{\gamma_n}_{jw}\mathbb{E}[\varepsilon_i^2]\mathbb{E}[\varepsilon_j^2]\varepsilon_h\varepsilon_l\varepsilon_m\varepsilon_w.
\end{equation*}
For $i \leq j$, tedious but straightforward calculations included for completeness yield
\begin{align*}
	&\mathbb{E}\left[\sum_{h = 1}^{i-1}\sum_{l = 1}^{i-1}\sum_{m = 1}^{j-1}\sum_{w = 1}^{j-1}P^{\gamma_n}_{ih}P^{\gamma_n}_{il}P^{\gamma_n}_{jm}P^{\gamma_n}_{jw}\varepsilon_h\varepsilon_l\varepsilon_m\varepsilon_w\right]\\
	 =& \sum_{l = 1}^{i-1}\sum_{h = 1}^{i-1}\sum_{m = 1}^{j-1}\sum_{w\neq m}^{j-1}P^{\gamma_n}_{il}P^{\gamma_n}_{ih}P^{\gamma_n}_{jm}P^{\gamma_n}_{jw}\mathbb{E}[\varepsilon_l\varepsilon_h\varepsilon_m\varepsilon_w] \\
	 & + \sum_{l = 1}^{i-1}\sum_{h\neq l}^{i-1}\sum_{m\neq l}^{j-1}P^{\gamma_n}_{il}P^{\gamma_n}_{ih}(P^{\gamma_n}_{jm})^2\mathbb{E}[\varepsilon_l]\mathbb{E}[\varepsilon_h\varepsilon_m^2]\\
	 & + \sum_{l = 1}^{i-1}\sum_{h\neq l}^{i-1}P^{\gamma_n}_{il}P^{\gamma_n}_{ih}(P^{\gamma_n}_{jl})^2\mathbb{E}[\varepsilon_l^3]\mathbb{E}[\varepsilon_h] + \sum_{l = 1}^{i-1}\sum_{m = 1}^{j-1}(P^{\gamma_n}_{il})^2(P^{\gamma_n}_{jm})^2\mathbb{E}[\varepsilon_l^2\varepsilon_m^2]\\
	 =& \sum_{l = 1}^{i-1}\sum_{h = 1}^{i-1}\sum_{m = 1}^{j-1}\sum_{w\neq m}^{j-1}P^{\gamma_n}_{il}P^{\gamma_n}_{ih}P^{\gamma_n}_{jm}P^{\gamma_n}_{jw}\mathbb{E}[\varepsilon_l\varepsilon_h\varepsilon_m\varepsilon_w] + \sum_{l = 1}^{i-1}\sum_{m = 1}^{j-1}(P^{\gamma_n}_{il})^2(P^{\gamma_n}_{jm})^2\mathbb{E}[\varepsilon_l^2\varepsilon_m^2]\\
	 =& \sum_{l = 1}^{i-1}\sum_{h\neq l}^{i-1}\sum_{m = 1}^{j-1}\sum_{w\neq m}^{j-1}P^{\gamma_n}_{il}P^{\gamma_n}_{ih}P^{\gamma_n}_{jm}P^{\gamma_n}_{jw}\mathbb{E}[\varepsilon_l\varepsilon_h\varepsilon_m\varepsilon_w]\\
	 &+ \sum_{l = 1}^{i-1}\sum_{m\neq l}^{j-1}\sum_{w\neq m}^{j-1}(P^{\gamma_n}_{il})^2P^{\gamma_n}_{jm}P^{\gamma_n}_{jw}\mathbb{E}[\varepsilon_l^2\varepsilon_w]\mathbb{E}[\varepsilon_m]\\
	 &+ \sum_{l = 1}^{i-1}\sum_{w\neq l}^{j-1}(P^{\gamma_n}_{il})^2P^{\gamma_n}_{jl}P^{\gamma_n}_{jw}\mathbb{E}[\varepsilon_l^3]\mathbb{E}[\varepsilon_w] + \sum_{l = 1}^{i-1}\sum_{m = 1}^{j-1}(P^{\gamma_n}_{il})^2(P^{\gamma_n}_{jm})^2\mathbb{E}[\varepsilon_l^2\varepsilon_m^2]\\
	 =& \sum_{l = 1}^{i-1}\sum_{h\neq l}^{i-1}\sum_{m = 1}^{j-1}\sum_{w\neq m}^{j-1}P^{\gamma_n}_{il}P^{\gamma_n}_{ih}P^{\gamma_n}_{jm}P^{\gamma_n}_{jw}\mathbb{E}[\varepsilon_l\varepsilon_h\varepsilon_m\varepsilon_w] + \sum_{l = 1}^{i-1}\sum_{m = 1}^{j-1}(P^{\gamma_n}_{il})^2(P^{\gamma_n}_{jm})^2\mathbb{E}[\varepsilon_l^2\varepsilon_m^2]\\
	 =& \sum_{l = 1}^{i-1}\sum_{h\neq l}^{i-1}\sum_{m\neq h}^{j-1}\sum_{w\neq m}^{j-1}P^{\gamma_n}_{il}P^{\gamma_n}_{ih}P^{\gamma_n}_{jm}P^{\gamma_n}_{jw}\mathbb{E}[\varepsilon_l\varepsilon_h\varepsilon_m\varepsilon_w]\\
	 &+ \sum_{l = 1}^{i-1}\sum_{h\neq l}^{i-1}\sum_{w\notin\{h, l\}}^{j-1}P^{\gamma_n}_{il}P^{\gamma_n}_{ih}P^{\gamma_n}_{jh}P^{\gamma_n}_{jw}\mathbb{E}[\varepsilon_l]\mathbb{E}[\varepsilon_w]\mathbb{E}[\varepsilon_h^2]\\
	 & + \sum_{l = 1}^{i-1}\sum_{h\neq l}^{i-1}P^{\gamma_n}_{il}P^{\gamma_n}_{ih}P^{\gamma_n}_{jh}P^{\gamma_n}_{jl}\mathbb{E}[\varepsilon_l^2]\mathbb{E}[\varepsilon_h^2]+ \sum_{l = 1}^{i-1}\sum_{m = 1}^{j-1}(P^{\gamma_n}_{il})^2(P^{\gamma_n}_{jm})^2\mathbb{E}[\varepsilon_l^2\varepsilon_m^2]\\
	 =& \sum_{l = 1}^{i-1}\sum_{h\neq l}^{i-1}\sum_{m\neq h}^{j-1}\sum_{w\neq m}^{j-1}P^{\gamma_n}_{il}P^{\gamma_n}_{ih}P^{\gamma_n}_{jm}P^{\gamma_n}_{jw}\mathbb{E}[\varepsilon_l\varepsilon_h\varepsilon_m\varepsilon_w] + \sum_{l = 1}^{i-1}\sum_{h\neq l}^{i-1}P^{\gamma_n}_{il}P^{\gamma_n}_{ih}P^{\gamma_n}_{jh}P^{\gamma_n}_{jl}\mathbb{E}[\varepsilon_l^2]\mathbb{E}[\varepsilon_h^2]\\
	 &+ \sum_{l = 1}^{i-1}\sum_{m = 1}^{j-1}(P^{\gamma_n}_{il})^2(P^{\gamma_n}_{jm})^2\mathbb{E}[\varepsilon_l^2\varepsilon_m^2]\\	 
	 =& \sum_{l = 1}^{i-1}\sum_{h\neq l}^{i-1}\sum_{m\neq h}^{j-1}\sum_{w \notin\{m, h\}}^{j-1}P^{\gamma_n}_{il}P^{\gamma_n}_{ih}P^{\gamma_n}_{jm}P^{\gamma_n}_{jw}\mathbb{E}[\varepsilon_l\varepsilon_h\varepsilon_m\varepsilon_w]\\
	 &+ \sum_{l = 1}^{i-1}\sum_{h\neq l}^{i-1}\sum_{m\notin\{h, l\}}^{j-1}P^{\gamma_n}_{il}P^{\gamma_n}_{ih}P^{\gamma_n}_{jm}P^{\gamma_n}_{jh}\mathbb{E}[\varepsilon_l]\mathbb{E}[\varepsilon_m]\mathbb{E}[\varepsilon_h^2]\\
	 &+ 2\sum_{l = 1}^{i-1}\sum_{h\neq l}^{i-1}P^{\gamma_n}_{il}P^{\gamma_n}_{ih}P^{\gamma_n}_{jl}P^{\gamma_n}_{jh}\mathbb{E}[\varepsilon_l^2]\mathbb{E}[\varepsilon_h^2] + \sum_{l = 1}^{i-1}\sum_{m = 1}^{j-1}(P^{\gamma_n}_{il})^2(P^{\gamma_n}_{jm})^2\mathbb{E}[\varepsilon_l^2\varepsilon_m^2]\\
	 =& \sum_{l = 1}^{i-1}\sum_{h\neq l}^{i-1}\sum_{m\neq h}^{j-1}\sum_{w \notin\{m, h\}}^{j-1}P^{\gamma_n}_{il}P^{\gamma_n}_{ih}P^{\gamma_n}_{jm}P^{\gamma_n}_{jw}\mathbb{E}[\varepsilon_l\varepsilon_h\varepsilon_m\varepsilon_w] \\
	 &+ 2\sum_{l = 1}^{i-1}\sum_{h\neq l}^{i-1}P^{\gamma_n}_{il}P^{\gamma_n}_{ih}P^{\gamma_n}_{jl}P^{\gamma_n}_{jh}\mathbb{E}[\varepsilon_l^2]\mathbb{E}[\varepsilon_h^2] + \sum_{l = 1}^{i-1}\sum_{m = 1}^{j-1}(P^{\gamma_n}_{il})^2(P^{\gamma_n}_{jm})^2\mathbb{E}[\varepsilon_l^2\varepsilon_m^2]\\
	 =& \sum_{l = 1}^{i-1}\sum_{h\neq l}^{i-1}\sum_{m\neq h}^{j-1}\sum_{w\notin\{m, h, l\}}^{j-1}P^{\gamma_n}_{il}P^{\gamma_n}_{ih}P^{\gamma_n}_{jm}P^{\gamma_n}_{jw}\mathbb{E}[\varepsilon_l\varepsilon_h\varepsilon_m]\mathbb{E}[\varepsilon_w]\\ &+\sum_{l=1}^{i-1}\sum_{h\neq l}^{i-1}\sum_{m\notin\{h,l\} }^{j-1}P^{\gamma_n}_{il}P^{\gamma_n}_{ih}P^{\gamma_n}_{jm}P^{\gamma_n}_{jl}\mathbb{E}[\varepsilon_l^2]\mathbb{E}[\varepsilon_h]\mathbb{E}[\varepsilon_m]\\
	 &+\sum_{l=1}^{i-1}\sum_{h\neq l}^{i-1}P^{\gamma_n}_{il}P^{\gamma_n}_{ih}(P^{\gamma_n}_{jl})^2\mathbb{E}[\varepsilon_l^3]\mathbb{E}[\varepsilon_h]
	 + 2\sum_{l = 1}^{i-1}\sum_{h\neq l}^{i-1}P^{\gamma_n}_{il}P^{\gamma_n}_{ih}P^{\gamma_n}_{jl}P^{\gamma_n}_{jh}\mathbb{E}[\varepsilon_l^2]\mathbb{E}[\varepsilon_h^2] \\
	 &+ \sum_{l = 1}^{i-1}\sum_{m = 1}^{j-1}(P^{\gamma_n}_{il})^2(P^{\gamma_n}_{jm})^2\mathbb{E}[\varepsilon_l^2\varepsilon_m^2]\\
	 =& 2\sum_{l = 1}^{i-1}\sum_{h\neq l}^{i-1}P^{\gamma_n}_{il}P^{\gamma_n}_{ih}P^{\gamma_n}_{jl}P^{\gamma_n}_{jh}\mathbb{E}[\varepsilon_l^2]\mathbb{E}[\varepsilon_h^2] + \sum_{l = 1}^{i-1}\sum_{m = 1}^{j-1}(P^{\gamma_n}_{il})^2(P^{\gamma_n}_{jm})^2\mathbb{E}[\varepsilon_l^2\varepsilon_m^2]\\
	 =& 2\sum_{l = 1}^{i-1}\sum_{h\neq l}^{i-1}P^{\gamma_n}_{il}P^{\gamma_n}_{ih}P^{\gamma_n}_{jl}P^{\gamma_n}_{jh}\mathbb{E}[\varepsilon_l^2]\mathbb{E}[\varepsilon_h^2] + \sum_{l = 1}^{i-1}\sum_{m\neq l}^{j-1}(P^{\gamma_n}_{il})^2(P^{\gamma_n}_{jm})^2\mathbb{E}[\varepsilon_l^2]\mathbb{E}[\varepsilon_m^2] \\
	 &+ \sum_{l = 1}^{i-1}(P^{\gamma_n}_{il})^2(P^{\gamma_n}_{jl})^2\mathbb{E}[\varepsilon_l^4]\\
	 =& 2\sum_{h = 1}^{i-1}\sum_{l\neq h}P^{\gamma_n}_{ih}P^{\gamma_n}_{il}P^{\gamma_n}_{jh}P^{\gamma_n}_{jl}\mathbb{E}[\varepsilon_h^2]\mathbb{E}[\varepsilon_l^2]\\
	 &+ \sum_{h = 1}^{i-1}(P^{\gamma_n}_{ih})^2(P^{\gamma_n}_{jh})^2\text{Var}[\varepsilon_h^2] + \sum_{h = 1}^{i-1}\sum_{l = 1}^{j-1}(P^{\gamma_n}_{ih})^2(P^{\gamma_n}_{jl})^2\mathbb{E}[\varepsilon_h^2]\mathbb{E}[\varepsilon_l^2],
\end{align*}
where $\text{Var}[\varepsilon_h^2] \coloneqq \mathbb{E}[\varepsilon_h^4] - \mathbb{E}[\varepsilon_h^2]\mathbb{E}[\varepsilon_h^2]$.
%

Defining $q \coloneqq \text{min}\{i,j\}$, the above display in turn implies
\begin{equation*}
	\begin{aligned}
	\mathbb{E}[\mathcal{V}_n^4] =& 32\sum_{i = 2}^n\mathbb{E}[\varepsilon_i^2]\sum_{j = 2}^n\mathbb{E}[\varepsilon_j^2]\sum_{h = 1}^{q-1}\sum_{l\neq h}P^{\gamma_n}_{ih}P^{\gamma_n}_{il}P^{\gamma_n}_{jh}P^{\gamma_n}_{jl}\mathbb{E}[\varepsilon_h^2]\mathbb{E}[\varepsilon_l^2]\\ 
	& + 16\sum_{i = 2}^n\mathbb{E}[\varepsilon_i^2]\sum_{j = 2}^n\mathbb{E}[\varepsilon_j^2]\sum_{h = 1}^{q-1}(P^{\gamma_n}_{ih})^2(P^{\gamma_n}_{jh})^2\text{Var}[\varepsilon_h^2]\\
	& + 16\sum_{i = 2}^n\mathbb{E}[\varepsilon_i^2]\sum_{j = 2}^n\mathbb{E}[\varepsilon_j^2]\sum_{h = 1}^{i-1}\sum_{l = 1}^{j-1}(P^{\gamma_n}_{ih})^2(P^{\gamma_n}_{jl})^2\mathbb{E}[\varepsilon_h^2]\mathbb{E}[\varepsilon_l^2]\\
	\equiv& \mathfrak{A}_n + \mathfrak{B}_n + s_n^4.
	\end{aligned}
\end{equation*}
In order to establish the convergence in Equation \eqref{Equation-Weird-Condition-2}, it is sufficient to show that $\frac{\mathfrak{A}_n + \mathfrak{B}_n}{s_n^4} \to 0$. To this end, note that
\begin{equation}
    \label{Equation-Frak-A}
	\begin{aligned}
	\mathfrak{A}_n &\leq C\sum_{i = 1}^n\sum_{j = 1}^n\sum_{h = 1}^{q-1}\sum_{l\neq h}P^{\gamma_n}_{ih}P^{\gamma_n}_{il}P^{\gamma_n}_{jh}P^{\gamma_n}_{jl}\\
	& = C\left(\sum_{i = 1}^n\sum_{j \neq i}\sum_{h = 1}^{q-1}\sum_{l\neq h}P^{\gamma_n}_{ih}P^{\gamma_n}_{il}P^{\gamma_n}_{jh}P^{\gamma_n}_{jl} + \sum_{i = 1}^n\sum_{h = 1}^{q-1}\sum_{l\neq h}(P^{\gamma_n}_{ih})^2(P^{\gamma_n}_{il})^2\right)\\
	&\leq C\left(\sum_{i = 1}^n\sum_{j < i}\sum_{h < j}\sum_{l\neq h}P^{\gamma_n}_{ih}P^{\gamma_n}_{il}P^{\gamma_n}_{jh}  P^{\gamma_n}_{jl} + r_n\right)\\
	& \leq  C\left(\sum_{i = 1}^n\sum_{j < i}\sum_{h < j}\sum_{l<h}P^{\gamma_n}_{ih}P^{\gamma_n}_{il}P^{\gamma_n}_{jh}  P^{\gamma_n}_{jl} + r_n\right)\\
	& \leq C\left(\left|\sum_{i = 1}^n\sum_{j < i}\sum_{h < j}\sum_{l<h}P^{\gamma_n}_{ih}P^{\gamma_n}_{il}P^{\gamma_n}_{jh}  P^{\gamma_n}_{jl}\right| + r_n \right)\\
	& \leq Cr_n
	\end{aligned}
\end{equation}
The first inequality follows from Assumption \ref{Assumption-MC}. The second inequality follows from Lemma \ref{Lemma-P-Linear-Algebra} \ref{Lemma-B1}. The third inequality follows from the symmetry of $P^{\gamma_n}$. The fifth inequality follows from Lemma \ref{Lemma-B2}. Next,
\begin{equation}
    \label{Equation-Frak-B}
    \begin{aligned}
        \mathfrak{B}_n &\leq C\sum_{i = 1}^n\sum_{j = 1}^n\sum_{h = 1}^{q-1}(P^{\gamma_n}_{ih})^2(P^{\gamma_n}_{jh})^2
                     \leq Cr_n,
    \end{aligned}
\end{equation}
the first inequality following from Assumption \ref{Assumption-MC} and the second from Lemma \ref{Lemma-P-Linear-Algebra} \ref{Lemma-B1}. Thus,
\begin{equation}
    \label{Equation-Weird-Condition-2-Limit}
    \begin{aligned}
	\frac{\mathfrak{A}_n + \mathfrak{B}_n}{s^4_n}  &\leq C \frac{r_n}{\left(\sum_{i = 1}^n\sum_{j\neq i}(P^{\gamma_n}_{ij})^2\right)^2}                   \leq 
	Cr^{1 - 2}_n
	           \to 0.
	\end{aligned}
\end{equation}
 where the second inequality follows from Equation \eqref{Equation-Assumption-High-Level-Implication} and the convergence from Assumption \ref{Assumption-Bekker} and Assumption \ref{Assumption-High-Level}. Equation \eqref{Equation-Weird-Condition-2-Limit} verifies the condition in Equation \eqref{Equation-Weird-Condition-2}, which in turn implies that the condition condition in Equation \eqref{Equation-Hall-Cond-2} holds.

Having verified the conditions in Equation \eqref{Equation-Hall-Cond-1} and Equation \eqref{Equation-Hall-Cond-2} we conclude that $RJAR^*_{\gamma_n} \overset{d}{\to}\mathcal{N}[0, 1]$. By Lemma \ref{Lemma-Variance-Cons}  $|\hat{\Phi}_{\gamma_n} - \Phi_{\gamma_n}|\overset{p}{\to} 0$ and hence also $\hat{\Phi}_{\gamma_n}/\Phi_{\gamma_n} \overset{p}{\to} 1$  under Assumptions \ref{Assumption-MC} and \ref{Assumption-High-Level}. Hence, under the null hypothesis, the continuous mapping theorem implies that $RJAR_{\gamma_n}(\beta_0) \overset{d}{\to} \mathcal{N}[0, 1]$.

\qed

\section{Simulation results with heteroskedastic errors \label{Appendix-Hetero}}
We extend the simulations in the main text to allow for heteroskedasticity in the error terms. As in the main text, the DGP is given by
\begin{subequations}
	\begin{align}
			y_i & = X_i\*\beta + \varepsilon_i\\
	        X_i & = Z_{i}'\*\pi + v_i, 
	\end{align}
\end{subequations}
for $i = 1, \dots, n = 100$. The IVs $Z_i$ are independent and identically Gaussian with mean 0 and $\text{Var}\left[Z_{il} \right] = 0.3 $ and $\text{Corr}\left[ Z_{il}, Z_{im}\right] = 0.5^{|l-m|}$. As in \citet{MavroeidisET}, the error terms are given by
\begin{equation*}
\begin{aligned}
    \varepsilon_i &= (\sigma_\varepsilon + ||Q_\varepsilon\tilde{Z}_i||_2)\eta_{1i}\\
    v_i &= (\sigma_v + ||Q_v\tilde{Z}_i||_2)\eta_{2i},
\end{aligned}
\end{equation*}
where $\tilde{Z}_i$ is the $4\times 1$ vector consisting of the first 4 elements of $Z_i$, $\sigma_\varepsilon = \sqrt 2$, $\sigma_v = 1$, $Q_\varepsilon$ and $Q_v$ are $4\times 4$ matrices controlling the degree of heteroskedasticity in the first and second stage, respectively,\footnote{The simulation setup in the main text corresponds to the case where  $Q_\varepsilon$ and $Q_v$ are set equal to the $4\times 4$ matrix of zeros.} and
\begin{equation*}
    \begin{aligned}
    [\eta_{1i}, \eta_{2i}]'
 &\sim \mathcal N \left[\left[\begin{smallmatrix}
     0\\
     0
 \end{smallmatrix}\right],   
\left[\begin{smallmatrix}
1   & 0.6 \\ 
0.6 & 1 
\end{smallmatrix}\right]
\right].
\end{aligned}
\end{equation*}
The~$[\eta_{1i}, \eta_{2i}]$ are mutually independent as well as independent of the~$\tilde{Z}_i$.
We set
\begin{equation}\label{eq:Qs}
    \begin{aligned}
        Q_\varepsilon &= \left[\begin{smallmatrix}
            2 & 0.8 & 0.6 & 0.4\\
            0.3  & 1.5 & 0.9 & 0.3\\
            0.8  & 0.6 & 1.9 & 0.2\\
            0.4  & 0.3 & 0.2 & 1.1
        \end{smallmatrix}\right] \quad \text{and} \quad Q_v = I_4,
    \end{aligned}
\end{equation}
following \citet{MavroeidisET}.\footnote{Our matrix $Q_\varepsilon$ corresponds to $Q_\varepsilon$ in \citet[Equation (4.9)]{MavroeidisET} with $\varrho$ (in their notation) set equal to 0.1.}

$\pi = \zeta\*\kappa$, where $\kappa$ is a vector of zeros and ones that varies with the type of DGP considered (sparse or dense, as modelled below), and $\zeta$ is some scalar that ensures that for a given value of $\mu^2$, the following relationship is satisfied:
\begin{equation*}
	\mu^2 = \frac{n\*\pi'\mathbb{E}\left[Z_{i}\*Z_{i}'\right]\*\pi}{\sigma_v^2}.
\end{equation*}

We note that in this heteroskedastic context, $\mu^2$ is only used to ensure that the coefficients on the IVs are the same as in the homoskedastic DGP considered in Section \ref{Section-RJAR-Simulations} in the main text, and does not measure identification strength.

As in the main text, we consider both a sparse first stage and a dense first stage. Sparsity in the first stage is modelled by setting $\kappa = [\iota_5', 0_{k_n-5}' ]'$, where $\iota_q$ is a $q\times 1$ vector of ones, and $0_q$ is a $q\times 1$ vector of zeros. Density in the first stage is modelled by setting $\kappa = [\iota_{0.4k_n}', 0_{0.6k_n}']'$. We consider $k_n = 30, 90, 190$. In the context of Assumption \ref{Assumption-High-Level}, we search over values greater than 1 when choosing~$\gamma^*_n$ in case~$r_n<k_n$. The variance estimator of MS occasionally yields a negative value. These cases are conservatively interpreted as a failure to reject the null hypothesis. As recommended by BCCH, $c_{BCCH} = 1.1$. As in the simulation section in CT, we set $\theta = 0.05$. The number of Monte Carlo replications is 10,000.

\subsection{Size}

Figure \ref{Figure-Size-Hetero} shows the simulation results with a sparse first stage for tests of size 0.01 to 0.99, that is the rejection frequency under $H_0:\beta_0 = 1$. The results are similar to the ones of the homoskedastic DGP considered in Section \ref{Section-RJAR-Simulations} in the main text, although the AR test of CT exhibts some mild size distortion.

\begin{figure}[H]
    \centering
    \begin{subfigure}[b]{0.5\textwidth}
        \includegraphics[width=\textwidth]{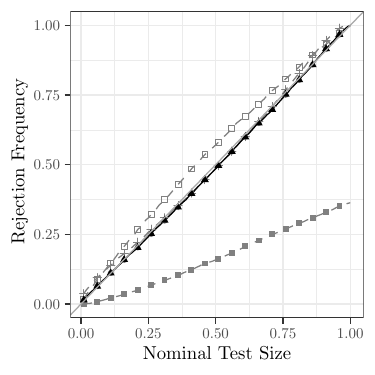}
        \caption{$k_n = 30$}
    \end{subfigure}%
    \begin{subfigure}[b]{0.5\textwidth}
        \includegraphics[width=\textwidth]{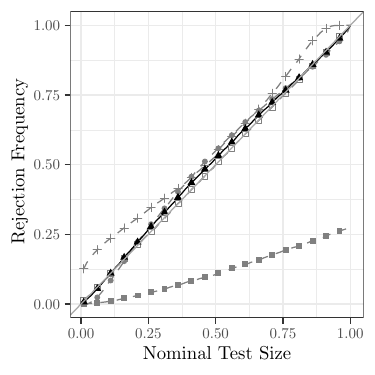}
        \caption{$k_n = 90$}
    \end{subfigure}
  \begin{subfigure}[b]{0.5\textwidth}
        \includegraphics[width=\textwidth]{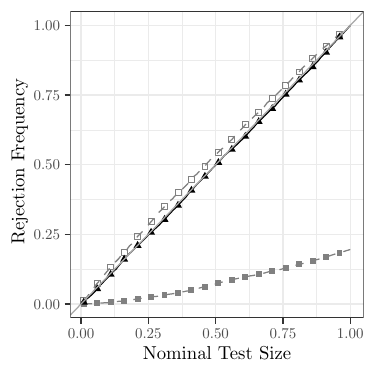}
        \caption{$k_n = 190$}
    \end{subfigure}
    
    \begin{subfigure}[c]{\textwidth}
    	\centering
    	\includegraphics{TablesFigures/PowerRevision/legend_ld.pdf}
    \end{subfigure}
    \caption{PP Plots for Sparse IVs, heteroskedastic errors, $\beta = 1$, $\mu^2 = 0$, $H_0: \beta_0 = 1$.}
    \label{Figure-Size-Hetero}
\end{figure}

\subsection{Power}

Figures \ref{Figure-Power-30-Hetero}--\ref{Figure-Power-190-Hetero} show the power of the tests when the number of IVs and the sparsity pattern of the first stage is varied. It is still the case that $H_0:\beta_0 = 1$. We report the results for $\mu^2 = 180$, as the power for $\mu^2 = 30$ is very low for all tests. The results confirm that also in this heteroskedastic setup, the RJAR test is as powerful as existing methods whenever these are applicable, and sometimes much more powerful.

\begin{figure}[H]
    \begin{subfigure}[b]{0.5\textwidth}
        \includegraphics[width=\textwidth]{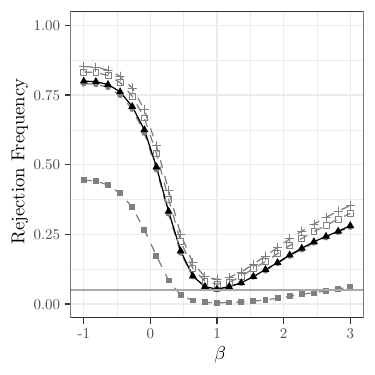}
        \caption{Sparse IVs}
    \end{subfigure}%
    \begin{subfigure}[b]{0.5\textwidth}
        \includegraphics[width=\textwidth]{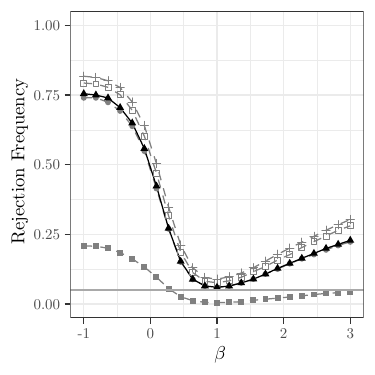}
        \caption{Dense IVs}
    \end{subfigure}
    \begin{subfigure}[c]{\textwidth}
    	\centering
    	\includegraphics{TablesFigures/PowerRevision/legend_ld.pdf}
    \end{subfigure}
    \caption{Power curves for 30 IVs and heteroskedastic errors. Nominal test size of 5\% indicated by the grey horizontal line. $\mu^2 = 180$, $H_0: \beta_0 = 1$.}
    \label{Figure-Power-30-Hetero}
\end{figure}

\begin{figure}[H]
    \begin{subfigure}[b]{0.5\textwidth}
        \includegraphics[width=\textwidth]{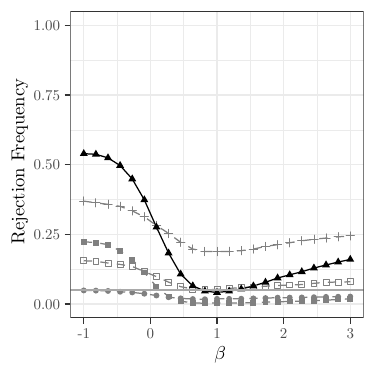}
        \caption{Sparse IVs}
    \end{subfigure}%
    \begin{subfigure}[b]{0.5\textwidth}
        \includegraphics[width=\textwidth]{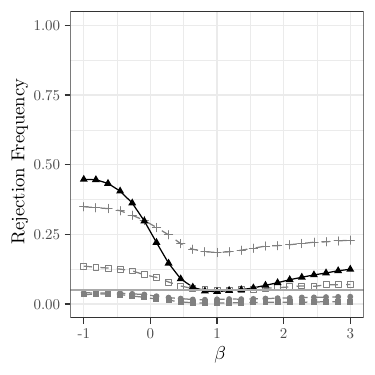}
        \caption{Dense IVs}
    \end{subfigure}
    \begin{subfigure}[c]{\textwidth}
    	\centering
    	\includegraphics{TablesFigures/PowerRevision/legend_ld.pdf}
    \end{subfigure}
    \caption{Power curves for 90 IVs and heteroskedastic error terms. Nominal test size of 5\% indicated by the grey horizontal line. $\mu^2 = 180$, $H_0: \beta_0 = 1$. }
    \label{Figure-Power-90-Hetero}
\end{figure}

\begin{figure}[H]
    \begin{subfigure}[b]{0.5\textwidth}
        \includegraphics[width=\textwidth]{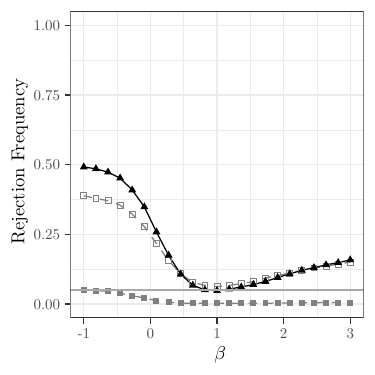}
        \caption{Sparse IVs}
    \end{subfigure}%
    \begin{subfigure}[b]{0.5\textwidth}
        \includegraphics[width=\textwidth]{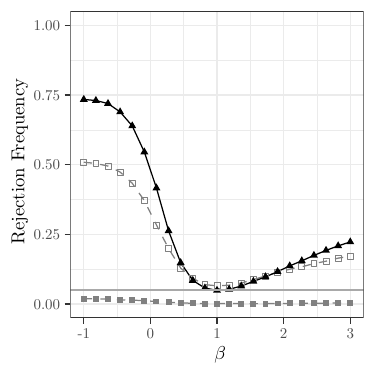}
        \caption{Dense IVs}
    \end{subfigure}
    \begin{subfigure}[c]{\textwidth}
    	\centering
    	\includegraphics{TablesFigures/PowerRevision/legend_hd.pdf}
    \end{subfigure}
    \caption{Power curves for 190 IVs and heteroskedastic error terms. Nominal test size of 5\% indicated by the grey horizontal line. $\mu^2 = 180$, $H_0: \beta_0 = 1$.}
    \label{Figure-Power-190-Hetero}
\end{figure}

\section{Simulation results for projection inference \label{Appendix-Projection}}

To study projection inference, we use the simulation setup for two endogenous variables of \citet{MavroeidisET}.

 The DGP is given by
\begin{equation*}
	\begin{aligned}
			y_i    & = \beta_1 X_{i1} + \beta_2 X_{i2} + \varepsilon_i\\
   	        X_{i1} & = Z_{i}'\*\pi_1 + v_{i1} \\
	        X_{i2} & = Z_{i}'\*\pi_2 + v_{i2},
	\end{aligned}
\end{equation*}
for $i = 1, \dots, n = 250$. Inference is conducted on $\beta_1$ by projecting out $X_{i2}$. $\beta_2 = 0$, $\pi_1 = \tilde{\iota}_{k_n}\zeta_1/(4n)^{1/2}$, $\pi_2 = \iota_{k_n}\zeta_2/(4n)^{1/2}$, $\tilde{\iota}_{k_n} = [\iota_{k_n/2}', -\iota_{k_n/2}']'$, $\zeta_1, \zeta_2 \in\{4, 40\}$ (weak and strong IVs), and $Z_i\overset{i.i.d.}{\sim}\mathcal{N}[0, I_{k_n}]$. The error terms are given by
\begin{equation*}
\begin{aligned}
    \varepsilon_i &= (||Q_\varepsilon\tilde{Z}_i||_F)\eta_{1i}\\
    v_{1i} &= (||Q_v\tilde{Z}_i||_F)\eta_{2i}\\
    v_{2i} &= (||Q_v\tilde{Z}_i||_F)\eta_{3i},\\
\end{aligned}
\end{equation*}
where $\tilde{Z}_i$ is the $4\times 1$ vector consisting of the first 4 elements of $Z_i$, the $4\times 4$ matrices $Q_\varepsilon$, $Q_v$ control the type and degree of heteroskedasticity and are set as in \eqref{eq:Qs}, and
\begin{equation*}
    \begin{aligned}
    [\eta_{1i}, \eta_{2i}, \eta_{3i}]'
 &\sim \mathcal N \left[\left[\begin{smallmatrix}
     0\\
     0\\
     0
 \end{smallmatrix}\right],   
\left[\begin{smallmatrix}
1   & 0.8 & 0.8 \\ 
0.8 & 1   & 0.3\\
0.8 & 0.3 & 1
\end{smallmatrix}\right]
\right].
\end{aligned}
\end{equation*}

We consider $k_n = 30$ (the number of IVs for which all approaches have correct size in the simulations reported in the main text) for three combinations of $\pi_1$ and $\pi_2$ (strong, strong), (weak, strong), (weak, weak). The variance estimator of MS occasionally yields a negative value. These cases are conservatively interpreted as a failure to reject the null hypothesis. $c_{BCCH}$, and $\theta$ are set as in the main text. The number of Monte Carlo replications is 10,000.

Figure \ref{Figure-Proj-Hetero} shows the power of the projection tests for $H_0:\beta_{1,0} = 0$. The tests have similar power except for the Sup Score test of BCCH, which has very low power. This is likely due to the dense parameterisation of the first-stage projection vectors in the DGP of \citet{MavroeidisET}. The AR test of CT exhibits some size distortions. This is likely due to the heteroskedasticity in the error terms.

\begin{figure}[H]
    \centering
    \begin{subfigure}[b]{0.5\textwidth}
        \includegraphics[width=\textwidth]{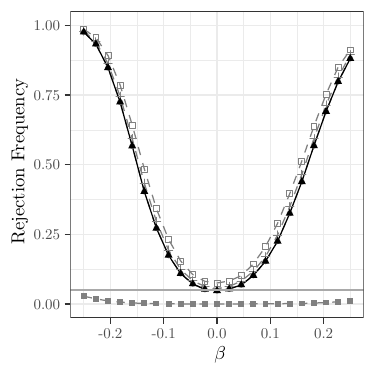}
        \caption{Strong ($\xi_1 = 40$, $\xi_2 = 40$)}
    \end{subfigure}%
    \begin{subfigure}[b]{0.5\textwidth}
        \includegraphics[width=\textwidth]{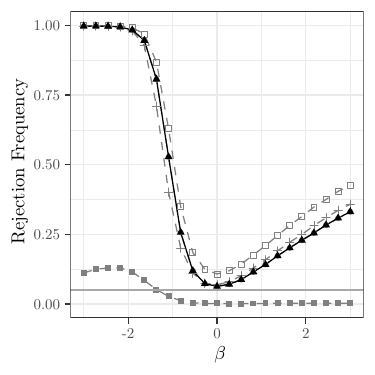}
        \caption{Mixed strength ($\xi_1 = 4$, $\xi_2 = 40$)}
    \end{subfigure}
    \begin{subfigure}[b]{0.5\textwidth}
    \centering
        \includegraphics[width=\textwidth]{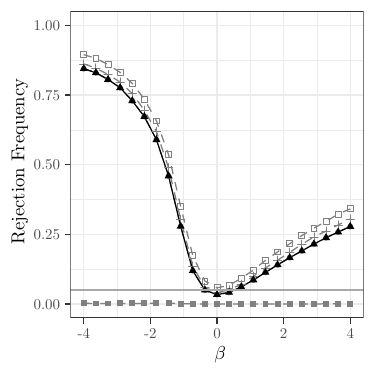}
        \caption{Weak ($\xi_1 = 4$, $\xi_2 = 4$)}
    \end{subfigure}
   
    \begin{subfigure}[c]{\textwidth}
    	\centering
   	\includegraphics{TablesFigures/PowerRevision/legend_ld.pdf}
   \end{subfigure}
   \caption{Power curves for projection inference with heteroskedastic error terms. Nominal test size of 5\% indicated by the grey horizontal line. $H_0: \beta_{1,0} = 0$. }
    \label{Figure-Proj-Hetero}
\end{figure}

\section{Simulation results with uncorrelated IVs \label{Appendix-Simulations-Uncorrelated-IVs}}

The same DGP as in the main text is used, except that $Z_i\overset{i.i.d.}{\sim}\mathcal{N}[0, I_{k_n}]$.

\begin{figure}[H]
    \begin{subfigure}[b]{0.5\textwidth}
        \includegraphics[width=\textwidth,height=0.9\textwidth]{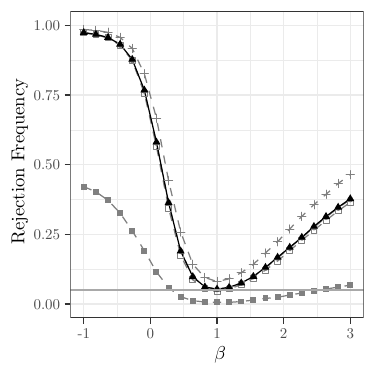}
        \caption{$\mu^2 = 30$, sparse IVs}
    \end{subfigure}%
    \begin{subfigure}[b]{0.5\textwidth}
        \includegraphics[width=\textwidth,height=0.9\textwidth]{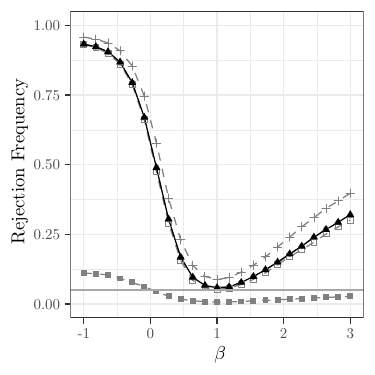}
        \caption{$\mu^2 = 30$, dense IVs}
    \end{subfigure}
    \begin{subfigure}[b]{0.5\textwidth}
        \includegraphics[width=\textwidth,height=0.9\textwidth]{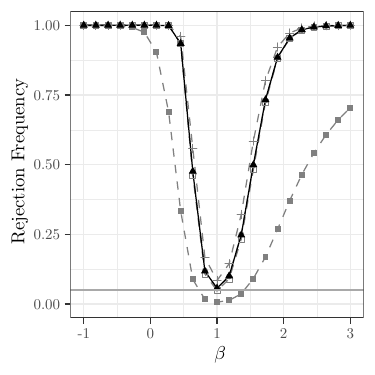}
        \caption{$\mu^2 = 180$, sparse IVs}
    \end{subfigure}%
    \begin{subfigure}[b]{0.5\textwidth}
        \includegraphics[width=\textwidth,height=0.9\textwidth]{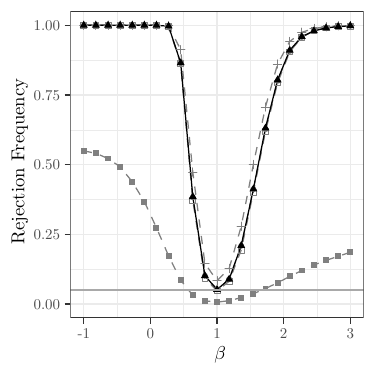}
        \caption{$\mu^2 = 180$, dense IVs}
    \end{subfigure}
    \begin{subfigure}[c]{\textwidth}
    	\centering
    	\includegraphics{TablesFigures/PowerRevision/legend_ld.pdf}
    \end{subfigure}
    \caption{Power curves for 30 independent IVs. Nominal test size of 5\% indicated by the grey horizontal line. $H_0: \beta_0 = 1$.}
    \label{Figure-Power-30-Homo-Indie}
\end{figure}

\begin{figure}[H]
    \begin{subfigure}[b]{0.5\textwidth}
        \includegraphics[width=\textwidth]{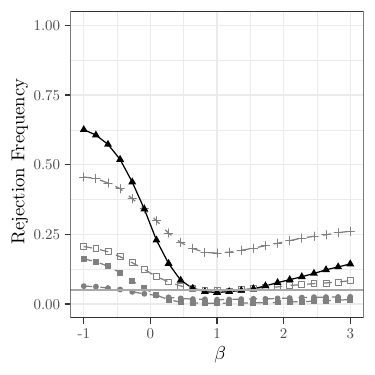}
        \caption{$\mu^2 = 30$, sparse IVs}
    \end{subfigure}%
    \begin{subfigure}[b]{0.5\textwidth}
        \includegraphics[width=\textwidth]{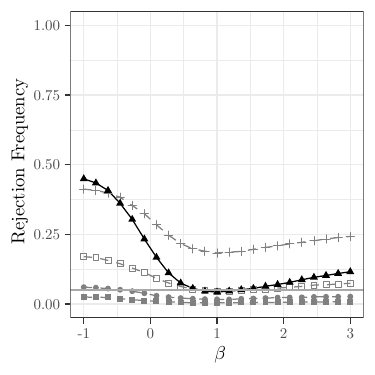}
        \caption{$\mu^2 = 30$, dense IVs}
    \end{subfigure}
    \begin{subfigure}[b]{0.5\textwidth}
        \includegraphics[width=\textwidth]{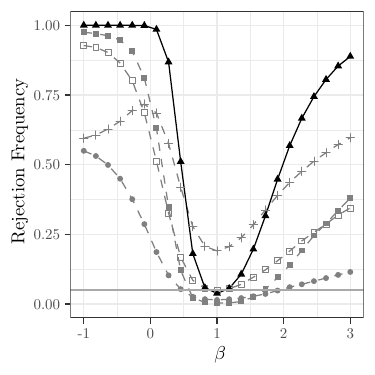}
        \caption{$\mu^2 = 180$, sparse IVs}
    \end{subfigure}%
    \begin{subfigure}[b]{0.5\textwidth}
        \includegraphics[width=\textwidth]{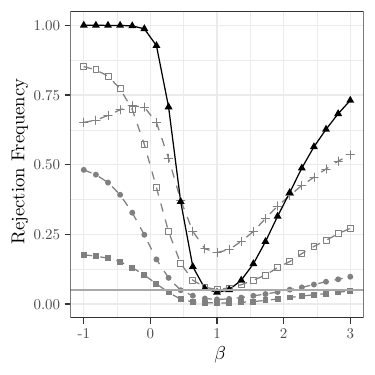}
        \caption{$\mu^2 = 180$, dense IVs}
    \end{subfigure}
    \begin{subfigure}[c]{\textwidth}
    	\centering
    	\includegraphics{TablesFigures/PowerRevision/legend_ld.pdf}
    \end{subfigure}
    \caption{Power curves for 90 independent IVs. Nominal test size of 5\% indicated by the grey horizontal line. $H_0: \beta_0 = 1$. }
    \label{Figure-Power-90-Homo-Indie}
\end{figure}

\begin{figure}[H]
    \begin{subfigure}[b]{0.5\textwidth}
        \includegraphics[width=\textwidth]{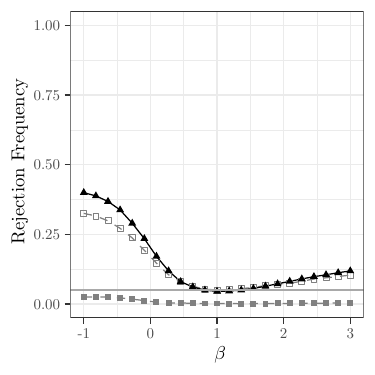}
        \caption{$\mu^2 = 30$, sparse IVs}
    \end{subfigure}%
    \begin{subfigure}[b]{0.5\textwidth}
        \includegraphics[width=\textwidth]{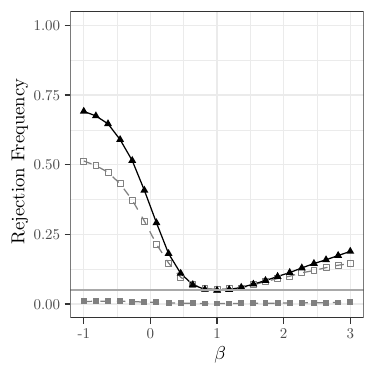}
        \caption{$\mu^2 = 30$, dense IVs}
    \end{subfigure}
    \begin{subfigure}[b]{0.5\textwidth}
        \includegraphics[width=\textwidth]{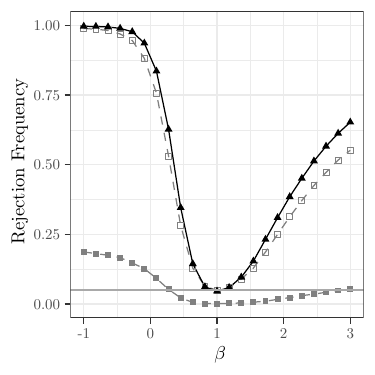}
        \caption{$\mu^2 = 180$, sparse IVs}
    \end{subfigure}%
    \begin{subfigure}[b]{0.5\textwidth}
        \includegraphics[width=\textwidth]{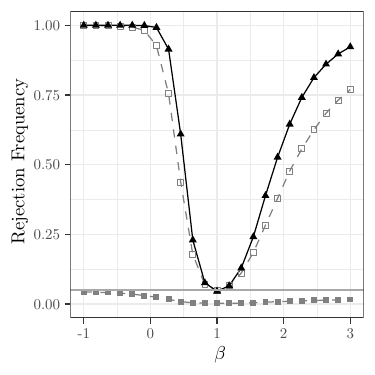}
        \caption{$\mu^2 = 180$, dense IVs}
    \end{subfigure}
    \begin{subfigure}[c]{\textwidth}
    	\centering
    	\includegraphics{TablesFigures/PowerRevision/legend_hd.pdf}
    \end{subfigure}
    \caption{Power curves for 190 independent IVs. Nominal test size of 5\% indicated by the grey horizontal line. $H_0: \beta_0 = 1$.}
    \label{Figure-Power-190-Homo-Indie}
\end{figure}

\section{Size of AR test of MS with $n = 1000$ and $k_n = 900$\label{Appendix-Simulations-1000}}

The same DGP as in the main text is used, except that $k_n = 900$ and $n = 1,000$.

\begin{figure}[H]
    \centering
        \includegraphics[width=0.5\textwidth]{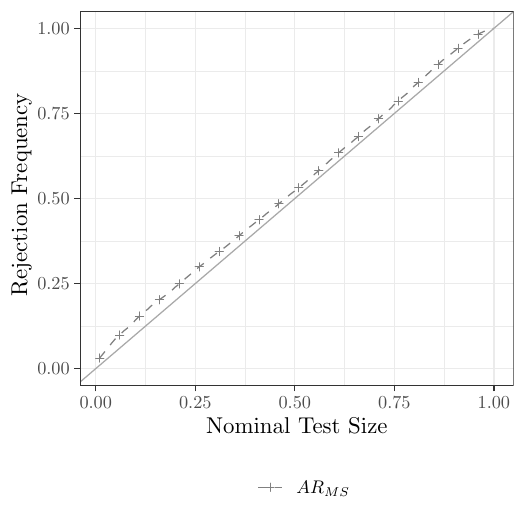}
    \caption{PP Plots for AR test of MS with sparse IVs, homoskedastic errors, $\beta = 1$, $\mu^2 = 0$, $H_0: \beta_0 = 1$, $n = 1000$, $k_n = 900$.}
    \label{Figure-Size}
\end{figure}


\end{document}